\documentclass[envcountsame,runningheads,orivec]{llncs}
\newcommand{\Jiwon}{\textsc{Cheong Ji-Won}\xspace}
\renewcommand{\Jiwon}{\raisebox{-0.5ex}{\includegraphics[height=2.5ex]{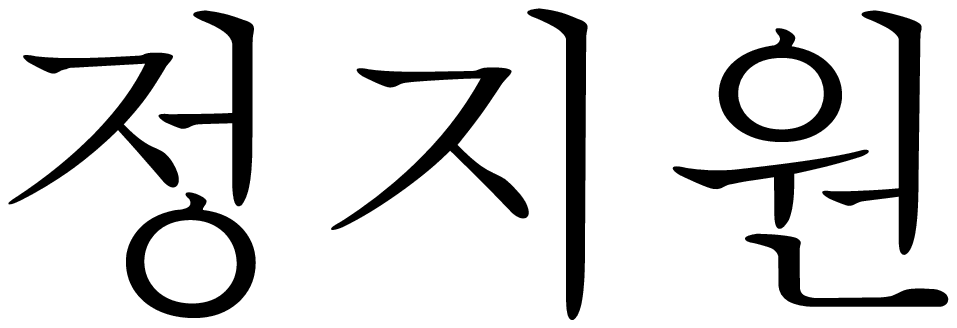}}\xspace}
\usepackage{a4wide}
\usepackage{amsmath}
\usepackage{amssymb}
\usepackage{xspace}
\usepackage{graphicx}
\newcommand{\IN}{\mathbb{N}}
\newcommand{\IR}{\mathbb{R}}

\newcommand{\Card}{\operatorname{Card}}
\newcommand{\VCdim}{\operatorname{VCdim}}
\newcommand{\loglog}{\operatorname{loglog}}
\newcommand{\polylog}{\operatorname{polylog}}
\newcommand{\VIS}{\mathcal{V}}
\newcommand{\Vis}{\operatorname{V}}
\newcommand{\VSP}{\operatorname{VSP}}
\newcommand{\Prob}{\operatorname{\textbf{Prob}}}
\newcommand{\calA}{\mathcal{A}}
\newcommand{\calL}{\mathcal{L}}
\newcommand{\calO}{\mathcal{O}}
\newcommand{\calNP}{\ensuremath{\mathcal{NP}}}

\newcommand{\calS}{\mathcal{S}}
\newcommand{\calT}{\mathcal{T}}
\newcommand{\calR}{\mathcal{R}}

\newcommand{\COMMENTED}[1]{}

\newcommand{\mycite}[2]{\cite[\textsc{#1}]{#2}}

\spnewtheorem{fact}[theorem]{Fact}{\bfseries}{\itshape}
\spnewtheorem{paradigm}[theorem]{Paradigm}{\bfseries}{\itshape}
\spnewtheorem{hypothesis}[theorem]{Hypothesis}{\bfseries}{\itshape}
\spnewtheorem{algorithm}[theorem]{Algorithm}{\bfseries}{\itshape}
\spnewtheorem{scholium}[theorem]{Scholium}{\bfseries}{\itshape}

\begin{document}
\title{Planar Visibility Counting\thanks{Supported
by DFG projects \texttt{Me872/12-1} within SPP 1307
and by \texttt{Zi1009/1-2}.\protect\newline
The last author is grateful to
\protect\Jiwon (n\'e \textsc{Otfried Schwarzkopf})
for the opportunity to visit \textsf{KAIST}.}}
\titlerunning{Planar Visibility Counting}
\author{M.~Fischer \and M.~Hilbig \and C.~J\"{a}hn \and F.~Meyer auf der Heide \and M.~Ziegler}
\institute{Heinz Nixdorf Institute
and University of Paderborn, 33095 GERMANY}
\authorrunning{M.~Fischer, M.~Hilbig, C.~J\"{a}hn, F.~Meyer auf der Heide, M.~Ziegler}
\date{}
\makeatletter
\renewcommand\maketitle{\newpage
  \refstepcounter{chapter}%
  \stepcounter{section}%
  \setcounter{section}{0}%
  \setcounter{subsection}{0}%
  \setcounter{figure}{0}
  \setcounter{table}{0}
  \setcounter{equation}{0}
  \setcounter{footnote}{0}%
  \begingroup
    \parindent=\z@
    \renewcommand\thefootnote{\@fnsymbol\c@footnote}%
    \if@twocolumn
      \ifnum \col@number=\@ne
        \@maketitle
      \else
        \twocolumn[\@maketitle]%
      \fi
    \else
      \newpage
      \global\@topnum\z@   
      \@maketitle
    \fi
    \thispagestyle{empty}\@thanks
    \def\\{\unskip\ \ignorespaces}\def\inst##1{\unskip{}}%
    \def\thanks##1{\unskip{}}\def\fnmsep{\unskip}%
    \instindent=\hsize
    \advance\instindent by-\headlineindent
    \if@runhead
       \if!\the\titlerunning!\else
         \edef\@title{\the\titlerunning}%
       \fi
       \global\setbox\titrun=\hbox{\small\rm\unboldmath\ignorespaces\@title}%
       \ifdim\wd\titrun>\instindent
          \typeout{Title too long for running head. Please supply}%
          \typeout{a shorter form with \string\titlerunning\space prior to
                   \string\maketitle}%
          \global\setbox\titrun=\hbox{\small\rm
          Title Suppressed Due to Excessive Length}%
       \fi
       \xdef\@title{\copy\titrun}%
    \fi
    \if!\the\tocauthor!\relax
      {\def\and{\noexpand\protect\noexpand\and}%
      \protected@xdef\toc@uthor{\@author}}%
    \else
      \def\\{\noexpand\protect\noexpand\newline}%
      \protected@xdef\scratch{\the\tocauthor}%
      \protected@xdef\toc@uthor{\scratch}%
    \fi
    \if@runhead
       \if!\the\authorrunning!
         \value{@inst}=\value{@auth}%
         \setcounter{@auth}{1}%
       \else
         \edef\@author{\the\authorrunning}%
       \fi
       \global\setbox\authrun=\hbox{\small\unboldmath\@author\unskip}%
       \ifdim\wd\authrun>\instindent
          \typeout{Names of authors too long for running head. Please supply}%
          \typeout{a shorter form with \string\authorrunning\space prior to
                   \string\maketitle}%
          \global\setbox\authrun=\hbox{\small\rm
          Authors Suppressed Due to Excessive Length}%
       \fi
       \xdef\@author{\copy\authrun}%
       \markboth{\@author}{\@title}%
     \fi
  \endgroup
  \setcounter{footnote}{\fnnstart}%
  \clearheadinfo}
\makeatother

\maketitle
\def\thefootnote{\fnsymbol{footnote}}\addtocounter{footnote}{1}
\begin{abstract}
For a fixed virtual scene (=collection of simplices) $\calS$
and given observer position $\vec p$, how many elements of
$\calS$ are weakly visible (i.e. not fully occluded by others) 
from $\vec p$?
The present work explores the trade-off between query time and preprocessing space
for these quantities in 2D: exactly, in the approximate deterministic,
and in the probabilistic sense. 
We deduce the \emph{existence} of an $\calO(m^2/n^2)$ space data structure 
for $\calS$ that,
given $\vec p$ and time $\calO(\log n)$, allows to approximate
the ratio of occluded segments up to arbitrary
constant absolute error; here $m$ denotes
the size of the \textsf{Visibility Graph}---which may be quadratic,
but typically is just linear in the size $n$ of the scene $\calS$.
On the other hand, we present a data structure 
\emph{constructible} in 
$\calO\big(n\cdot\log(n)+m^2\cdot\polylog(n)/\ell\big)$
preprocessing time and space
with similar approximation properties 
and query time $\calO(\ell\cdot\polylog n)$,
where $1\leq\ell\leq n$ is an arbitrary parameter.
We describe an implementation of this approach 
and demonstrate the practical benefit of the 
parameter $\ell$ to trade memory for query time
in an empirical evaluation
on three classes of benchmark scenes.
\end{abstract}
\begin{minipage}[c]{0.95\textwidth}
\renewcommand{\contentsname}{}
\tableofcontents
\end{minipage}
\stepcounter{footnote}
\section{Motivation and Introduction}
Back in the early days of computer graphics, \textsf{hidden surface removal}
(and \textsf{visible surface calculation})
was a serious computational problem: 
for a fixed virtual 3D scene and given observer position, 
(partition and) select those scene primitives
which are (and are fully) visible to the observer.
Because of its importance, this problem
has received considerable scientific attention with many suggestions
of deep both combinatorial and geometric algorithms 
for its efficient solution.
The situation changed entirely when the (rather unsophisticated)
\textsf{z-buffer} algorithm became available in common consumer graphics cards:
with direct hardware support and massive parallelism (one gate per pixel),
it easily outperforms software-based approaches with their 
(usually huge factors hidden in) asymptotic big-Oh running times
\cite{McKenna}.
For a fixed resolution, the z-buffer can render scenes of $N$
triangles on-line in time essentially linear in $N$ with a small constant.
However even this may be too slow in order to visualize virtual worlds 
consisting of several hundreds of millions of triangles at
interactive frame rates. Computer graphics literature is
filled with suggestions of how to circumvent this problem;
for example by approximating (in some intuitive, informal sense) 
the observer's views. Here, the benefit of a new algorithm
is traditionally demonstrated by evaluating it, and
comparing it to some previous `standard' algorithm,
on few `standard' benchmark scenes and on selected hardware.
We on the other hand are interested in algorithms with
provable properties, and to this end restrict to

\subsection{Conservative Occlusion Culling} \label{s:Culling}

\begin{definition} \label{d:culling}
Objects which are hidden to the observer behind (possibly a collection of) other
objects \emph{may}, but need not, be filtered from the stream sent to 
the rendering hardware,
whereas any at least partially visible object \emph{must} be visualized.
\end{definition}
Here, ``conservative'' reflects that the 
rendering algorithm must not affect the visual apprearance
compared to the bruce-force approach of sending all objects
to the hardware. Occlusion culling can speed up the 
visualization particularly of very large scenes (e.g.
virtual worlds as in \textsf{Second Life} or
\textsf{World of Warcraft}) where, composed from
literally billions of triangles, typically `just'
some few millions are actually visible at any instant.
Other scenes, or viewpoints within a scene, 
admit no sensible occlusion; for instance
the leaves of a virtual forest naturally do not fully
screen sight to the sun or to each other, similarly
for CAD scenes of lattice or similar constructions.
In such cases, spending computational efforts on occlusion culling
is futile and actually bound to a net performance \emph{loss}.
Between those extremes, and particularly for an observer
moving between occluded and free parts of a large scene,
the algorithmic overhead of 
more or less thoroughly filtering out hidden primitives
generally trades off against the benefit in
reduced rendering complexity. Put differently:
Graphics hardware taking care of the visibility problem
anyway opens the chance to hybridize with software performing 
either coarse (and quick) or careful (and slow) culling
and leave the rest to the z-buffer.

\subsection{Adaptive Occlusion Culling}
It is our purpose is to explore this trade-off 
and to make algorithms adapt to each 
specific virtual scene and observer position
in order to exploit it in a
well-defined and predictable way. To this end
we propose so-called \emph{visibility counts}
(the number of primitives weakly visible from a
given observer position) as a quantitative measure
of how densely occluded a rendering frame is, and whether and by how much 
occlusion culling therefore can be, expected to pay off.
For technical reasons employed in Sections~\ref{s:Tradeoff} and later, 
the formal notion slightly more generally
captures the visibility of `target' scenes
through `occluder' scenes:

\begin{definition}[visibility count] \label{d:Count}
For a scene $\calS=\{S_1,\ldots,S_n\}$ of
`geometric primitives' $S_i\subseteq\IR^d$,
a subset of `\emph{targets}' $\calT\subseteq\calS$,
and an observer position $\vec p\in\IR^d$, let
\[ \VIS(\calS,\vec p,\calT)\;:=\;\big\{
T\in\calT\big| \exists \vec q\in T: \forall S\in\calS\setminus\{T\}: 
 [\vec p,\vec q]^\circ\cap S=\emptyset \big\} \]
and denote by $\Vis(\calS,\vec p,\calT):=\Card\VIS(\calS,\vec p,\calT)$
the number of objects in $\calT$ 
weakly visible (i.e. not fully occluded)
from $\vec p$ through $\calS$.
Here, $[\vec p,\vec q]^\circ:=\{\lambda\cdot\vec p+(1-\lambda)\cdot\vec q:0<\lambda<1\}$
denotes the (relatively open)
straight line segment connecting $\vec p$ and $\vec q$.
\end{definition}
For scenes $\calS$ and observer positions $\vec x$
with $\Vis(\calS,\vec x,\calS)\ll\Card\calS$,
occlusion culling is likely to pay off;
whereas for $\Vis(\calS,\vec x,\calS)\cong\Card\calS$
it is not. Quantitatively we have the following
\begin{hypothesis} \label{h:Count}
Each culling algorithm $\calA$ can be assigned 
a threshold function $\theta_{\calA}(n)\in[0,1]$ 
such that, for scenes $\calS$ and observer positions $\vec x$
with visibility \emph{ratios}
$\Vis(\calS,\vec x,\calS)/\Card(\calS)$
(significantly) beyong $\theta_{\calA}(n)$,
it yields a net rendering benefit and 
(significantly) below does not.
\end{hypothesis}

\subsection{Combinatorial Geometry and Randomized Computation} \label{s:Provable}
Adaptivity constitutes an important issue in Computational Geometry;
for instance in the context of \textsf{Range Searching} problems
whose running time is preferably \emph{output sensitive},
i.e. of the form $\calO\big(f(n)+k\big)$ where $n$ denotes
the overall number of objects and $k$ those that are actually
reported; compare \cite{Filtering}.

Adaptivity is of course a big topic in computer graphics as well.
However this entire field, driven by the impetus to quickly 
visualize (e.g. at 20fps) 
concrete scenes in newest interactive video games,
generally focuses on innovative heuristics and techniques
at a tremendous pace.
We on the other hand are interested in algorithms with
provable properties based on formal and sound analyses
and in particular
with respect to well-defined measures of adaptivity.
This of course calls for an application of
computational and combinatorial geometry
\cite{deBerg,Edelsbrunner}.

\begin{paradigm}[Computational Geometry in Computer Graphics] \label{p:Paradigm}
For interactive visualization of very large virtual scenes
of size $N\gtrsim 10^7$, algorithms must run in 
sublinear time $\calO(N^{\alpha})$, $\alpha<1$,
using preprocessed data structures 
of almost linear space $\calO(N^{1+\epsilon})$,
$\epsilon\ll1$, provably!
\end{paradigm}
Here (time and) space complexity refers to the
number of (operations on) unit-size real coordinates 
used (performed) by an algorithm---as opposed to, e.g., 
rationals of varying bitlength\footnote{See however
Item~d) in Section~\ref{s:Perspective} below}.
Also visibility is considered in the geometric sense
(as opposed to e.g. pixel-based notions):
point $\vec q$ is visible from observer position $\vec p$
if both can be connected by an ideal light ray
(=straight-line segment not intersecting any other
part of the scene), recall Definition~\ref{d:Count}.
Our algorithm features a parameter $1\!\leq\!\ell\!<\!n$ to trade
preprocessing space for query time.

Randomized algorithms are quite common in computer science for
their efficiency and implementation simplicity. They have also
entered the field of computer graphics. Here these
techniques are employed to render only a small random sample 
of the (typically very large) scene in such a way
that it appears similar to the entire scene
\cite{Wand1,Wand2,visibilitysampling}.
Our goal, on the other hand, is to approximate
the \emph{count} of visible objects (Definition~\ref{d:Count}),
not their appearance.

\subsection{Visibility} \label{s:Visibility}
Visibility comprises a highly active field of research,
both heuristically 
and in the sound framework of computational geometry \cite{VisibilitySurvey}.
Particularly the latter has proven combinatorially and algorithmically 
non-trivial already in the plane \cite{ORourke,PlanarVisibility}.
Here the case of (simple) polygons is well studied \cite{Charneau};
and so is point--point, point--segment, and segment--segment visibility
for scenes $\calS$ of $n$ non-crossing line segments,
captured e.g. in the \textsf{Visibility Graph} data structure \cite{VisibilityGraph}.
Its nodes correspond either to segments or to 
segment endpoints (or to both: a bipartite graph); 
and two nodes get joined
by an edge if one can partly see the other.
Weak segment--segment visibility for instance 
amounts to the $\calO(n^2)$ questions 
(namly for each pair of segments $A$ and $B$)
of whether there exist points $\vec a\in A$ and
$\vec b\in B$ such that $\vec a$ is visible from $\vec b$.

\medskip
We, too, ask for weak segment visibility;
however in our case 
the observer is not restricted to positions on
segments of the scene but may move freely between
them. For instance we shall want to efficiently calculate
visibility counts for singleton targets
$\Vis(\calS,\vec x,\{T\})$:

\begin{problem} \label{p:Visibility}
Fix a collection $\calS$ of non-intersecting segments 
in the plane and one further segment $T$.
Preprocess $(\calS,T)$ into an almost linear 
(or merely worst-case subquadratic) 
size data structure such as
to decide in sublinear time queries of the following type: \\
Given $\vec x\in\IR^2$, is $T$ (partly) visible through $\calS$ ?
\end{problem}
Sections~\ref{s:IntervalTree}, \ref{s:RotSweep}, and \ref{s:VSP} recall
two algorithms that meet either the space or the time
requirement but not both.

\subsection{Overview} \label{s:Overview}
An empirical verification of Hypothesis~\ref{h:Count} in dimensions
2, $2\tfrac{1}{2}$, and $3$,
is the subject of a separate work \cite{SEA}.
Our aim here is to explore the complexity of calculating 
the visibility counts, thus providing rendering algorithms
with the information for deciding whether to cull or not.
In view of the large virtual scenes and the high frame rates 
required by applications, we have to consider both 
computational resources, query time and preprocessing space,
simultaneously. Section~\ref{s:Exact} focuses on 
the problem of calculating visibility counts exactly,
mostly based on the \textsf{Visibility Space Partition};
our main result here is a preprocessing algorithm
with output-sensitive running time.
Section~\ref{s:Approximate} weakenes the problem
to approximate calculations:
first showing the \emph{existence} of a rather
small data structure with logarithmic query time
in Section~\ref{s:Coarse}.
However this data structure seems hard to
\emph{construct} in reasonable time, 
therefore Sections~\ref{s:Chernoff}ff 
consider approaches based on random sampling.
Section~\ref{s:Evaluation} describes an implementation
and evaluation of this algorithm.

\section{Exact Visibility Counting} \label{s:Exact}
This section recalls combinatorial worst-case approaches for calculating
visibility counts according to Definition~\ref{d:Count}.
Many efficient algorithms are known for visibility \emph{reporting}
problems, that is for determining the \emph{view} of an observer
\cite{Flatland2}; however since reporting may involve output of
linear size, such aproaches are generally inappropriate for our goal 
of \emph{counting} in \emph{sub}linear time. On the other hand,
logarithmic time becomes easily feasible when permitting
quartic space in the worst-case based on the Visibility Space Partition (VSP). 
The main result of this section, Theorem~\ref{t:OutputSensitive} 
yields an output-sensitive time 
algorithm for computing the VSP of a given set of line segments in the plane.

\subsection{Reverse Painter's Algorithm} \label{s:IntervalTree}
Prior to the hardware z-buffer, 
\textsf{Painter's Algorithm} was sometimes considered
as a means to hidden surface elimination (at least in the
2D case): Draw all objects in back-to-front order,
thus making closer ones paint over (and thus correctly cover)
those further away. This of course relies on being
able to efficiently find such an order: which is easily
seen impossible in general unless we `cut' some objects. 
Now two-dimensional \textsf{BSP Trees} 
provide a means to find such an order and a way to cut
objects appropriately 
without increasing the overall size too much. We report
from \mycite{Section~12}{deBerg}:

\begin{fact}
Given a collection $\calS$ of $n$ non-crossing line segments in the plane,
a BSP Tree of $\calS$ can be constructed in time
and space $\calO(n\cdot\log n)$.
\end{fact}
Now instead of drawing the cut segments in back-to-front order
(relative to the observer),
feeding them into an \textsf{Interval Tree} in front-to-back order
reveals exactly which of them are weakly visible and which not.
Since insertion into an Interval Tree of size $n$ takes time
$\calO(n\cdot\log n)$ we conclude \cite{Sequin}:

\begin{lemma} \label{l:IntervalTree}
Given a collection $\calS$ of $n$ non-crossing line segments in the plane
and an observer position $\vec x\in\IR^2$,
$\Vis(\calS,\vec x,\calS)$ can be calculated in
time $\calO(n\cdot\log^2 n)$ and space $\calO(n\cdot\log n)$.
\end{lemma}
Notice that preprocessing $\calS$ into a BSP Tree
accelerates the running time `only' by a constant factor.

\subsection{Rotational Sweep} \label{s:RotSweep}
One can improve Lemma~\ref{l:IntervalTree} by a
logarithmic factor:

\begin{lemma} \label{l:RotSweep}
Given a collection $\calS$ of $n$ non-crossing line segments in the
plane and an observer position $\vec x$,
$\Vis(\calS,\vec x,\calS)$ can be calculated in
time $\calO(n\cdot\log n)$ and space $\calO(n)$.
\end{lemma}
\begin{proof}{Sketch}
First mark all segments \texttt{invisible}.
Then consider the $2n$ endpoints of $\calS$ in angular order
around $\vec x$ while keeping track of the order of the 
segments according to their proximity to the observer,
the closest one thus being \texttt{visible}:
whenever a new segments starts insert it
into an appropriate data structure
in time $\calO(\log n)$,
whenever one ends remove it.
Since the initial sorting also takes time $\calO(n\cdot\log n)$,
we remain within the claimed bounds.
\qed\end{proof}
Nevertheless the running time still fails to meet
Paradigm~\ref{p:Paradigm}.
Also, these approaches seem to offer
no way to take advantage of a singleton target 
for the purpose of Problem~\ref{p:Visibility}.

\subsection[Visibility Space Partition (VSP)]{Visibility Space Partition} \label{s:VSP}
Lemmas~\ref{l:IntervalTree} and \ref{l:RotSweep} work without any,
and do not benefit asymptotically from, 
preprocessing of the fixed scene $\calS$.
On the other hand by the so-called \textsf{locus approach}---storing 
\emph{all} visibility counts in a \textsf{Visibility Space Partition} 
(VSP)---they can later be recovered in logarithmic running time
\cite{AspectSurvey}:

\begin{lemma} \label{l:VSP}
\begin{enumerate}
\item[a)]
For a collection $\calS$ of 
$n$ non-crossing line segments in the plane,
there exists a partition of $\IR^2$ into 
$\calO(n^4)$ convex cells such that,
for all observer positions $\vec x\in C$ within one cell $C$, 
$\VIS(\calS,\cdot,\calS)$ is the same.
\item[b)]
The data structure indicated in a) and including 
for each cell its corresponding visibility count 
$\Vis(\calS,C,\calS)$
uses storage $\calO(n^4\cdot\log n)$ and
can be computed in time $\calO(n^5\cdot\log n)$.
Then given an observer position $\vec x$,
its corresponding cell $C$, and the associated 
visibility count, can be identified in time $\calO(\log n)$.
\item[c)]
When charging only real data and operations
(more specifically: If an $n$-bit string is considered 
to occupy one memory cell and the union of two of them
computable within one step), the above data structure
including for each cell its visibility $\VIS(\calS,C,\calS)$
uses storage $\calO(n^4)$ and can be calculated
in time $\calO(n^4\cdot\log^2n)$.
\item[d)]
Item~a) extends to the case
of $(d-1)$-simplices in $d$-dimensional space
in that the number of convex cells with equivalent 
observer visibility can be bounded by 
$\calO(n)^{d^2}$.
\end{enumerate}
\end{lemma}
\begin{proof}
\begin{enumerate}
\item[a)]
Draw lines through all $\binom{2n}{2}$ 
pairs of the $2n$ segment endpoints. 
It is easy to see that, in order for a near segment
to appear in sight, the observer has to cross 
one of these $\calO(n^2)$ lines; 
compare Lemma~\ref{l:TimeVersusSpace} below.
Hence, within each of the $\calO(n^4)$ cells
they induce, the subset of segments weakly visible
remains the same; compare Figure~\ref{f:aspect3}.
\item[b)]
$\calO(n^2)$ lines induce an arrangement of
overall complexity, and can be constructed
in time $\calO(n^4)$ \mycite{Section~8.3}{deBerg}.
The visibility
number associated with each cell is bounded
by $n$ and hence can be stored using
$\calO(\log n)$ bits; its calculation
according Lemma~\ref{l:RotSweep}
takes time $\calO(n\cdot\log n)$ each.
Finally, the planar subdivision induced by the $\calO(n^4)$
edges of the arrangement can be turned into
a data structure supporting point-location in
$\calO(\log n)$ \mycite{Theorem~6.8}{deBerg}.
\item[c)]
In the proof to b), constructing the arrangement
(i.e. the planar partition into cells)
and determining the visibility count of each cell
were two separate steps which we now merge
using divide-and-conquer:
In the first phase
calculate the VSP of the first two segments of $\calS$,
then that of the next, and so on; in each VSP store, for
each cell, the visibility vector, i.e. the 0/1 bitstring
recording which segments from $\calS$ are visible (1) and
which are not (0). 
In the next phase
overlay the first two VSPs of two segments into 
one of the first four segments, and store for each refined
cell the union of the 0/1 bitstrings: thus 
keeping track of its visibility;
similarly for the next two VSPs of the next four segments.
Then proceed to VSPs of eight segments each; and so on.
We therefore have $\calO(\log n)$ phases;
and, according to \mycite{Theorem~2.6}{deBerg},
the last (as well as each previous) phase 
takes time $\calO(n^4\cdot\log n)$.
\item[d)]
Similarly to the proof of a), consider all
$dn$ vertices of the $n$ simplices.
Any $d$-tuple of them induces a 
hyperplane; and a change in sight requires
the observer to cross some of these 
$N:=\binom{nd}{d}\leq \calO(n)^d$
hyperplanes. 
$N$ hyperplanes in $d$-space induce
an arrangement of complexity
$\calO(N)^{d}$ \cite{Edelsbrunner}.
\qed
\end{enumerate}
\end{proof}

\begin{figure}[htb]\centerline{%
\includegraphics[width=0.7\textwidth,height=0.3\textwidth]{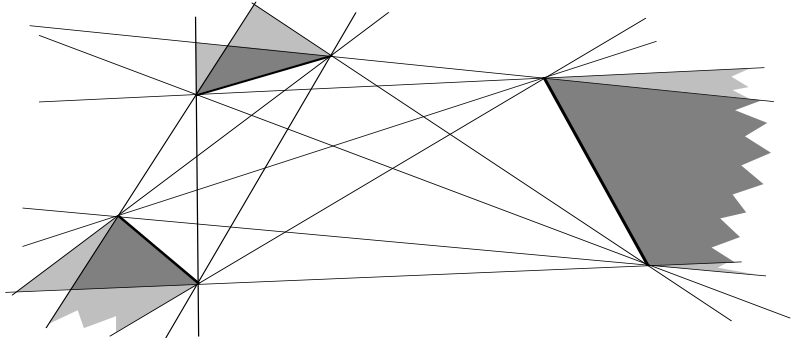}}%
\caption{\label{f:aspect3}Visibility Space Partition
of three segments: an observer in the dark gray area
can see exactly one segment, in the light gray area
exactly two, and otherwise all three of them.}
\end{figure}

\subsection{Size of Visibility Space Partitions} \label{s:VSPsize}
Lemma~\ref{l:VSP}a+b) bounds the 
size of the VSP data structure
by order $n^4$. It turns out that this
bound is sharp in the worst case---but
not for many `realistic' examples.
This is due to many of the $\Theta(n^2)$ lines
employed in the proof of Lemma~\ref{l:VSP}a)
inducing unnecessarily fine subdivisions
of viewpoint space. In Figure~\ref{f:VSPlines}a)
for instance, the dotted parts are dispensible.

In order to avoid trivialities, we want to restrict 
to \emph{non}degenerate segment configurations $\calS$.
However this notion is subtle because the lines 
induced by $\calS$ defining the VSP typically \emph{are}
degenerate: many (more than two) of them meet in one
common (segment end) point.

\begin{definition} \label{d:Nondegenerate}
A family $\calS$ of segments in the plane is \textsf{nondegenerate} if 
\begin{enumerate}
\item[i)] any two segments meet only in their common endpoints.
\item[ii)] No three endpoints share a common line;
\item[iii)] Any two lines, defined by pairs of endpoints, do meet.
\end{enumerate}
\end{definition}
We have already referred to (and implicitly
employed in Lemma~\ref{l:VSP} a refinement of)
the Visibility Space Partition; so here finally
comes the formal

\begin{definition} \label{d:VSP}
For two non-degenerate collections $\calS$ and $\calT$
of segments in the plane, partition all viewpoints $\vec p\in\IR^2$ 
into classes having equal visibility $\VIS(\calS,\vec p,\calT)$.
Moreover let $\VSP(\calS,\calT)$ denote the collection
of connected components of these equivalence classes.
The \textsf{size} of $\VSP$ is the number of line segments
forming the boundaries of these components.
\end{definition}
Observe that
$\VSP(\calS,\calT)$ indeed constitutes a planar subdivision:
a coarsening of the $\calO(n^4)$ convex polygons
induced by the arrangement of $\calO(n^2)$ lines from
the proof of Lemma~\ref{l:VSP}a).
In fact a class of viewpoints of equal visibility 
can be disconnected and delimited by very many segments,
hence merely counting the number of classes or cells
does not reflect the combinatorial complexity.
Lemma~\ref{l:VSP}a) and Lemma~\ref{l:VSPsize}~a)
correspond to \mycite{Exercise~6.1.7}{Matousek}.

\begin{lemma} \label{l:VSPsize}
\begin{enumerate}
\item[a)]
Even for a singleton target $T$, there exist
a nondegenerate line segment configurations $\calS$ such that
$\VSP(\calS,\{T\})$ has $\Omega(n^4)$ separate
connected components.
\item[b)]
To each $n$,
there exists a nondegenerate configuration $\calS$
of at least $n$ segments admitting a convex planar subdivision
of complexity $\calO(n)$ such that, from within each cell,
the view to $\calS$ is constant; i.e.
$\VSP(\calS,\calS)$ has linear size.
\item[c)]
The size of $\VSP(\calS,\calS)$ is at most quadratic
in the size $m$ of the \textsf{Visibility Graph} 
of $\calS$ (recall Section~\ref{s:Visibility}).
\item[d)]
A data structure as in Lemma~\ref{l:VSP} can be 
calculated in time $\calO(n\cdot\log n+m^2\cdot\log^2n)$
and space $\calO(m^2)$.
\end{enumerate}
\end{lemma}
Since the Visibility Graph itself can have at most
quadratically more edges than vertices, Item~c)
strengthens Lemma~\ref{l:VSP}a).
Empirically we have found that a `random' scene
typically induces a VSP of roughly quadratic size.
This agrees with a `typical' scene to have a 
linear size Visibility Graph according to
\cite{ExpectedSize}.
\begin{proof}[Lemma~\ref{l:VSPsize}]
\begin{enumerate}
\item[a)]
Figure~\ref{f:ORourke} is a small modification
of \mycite{Fig.~8.13}{ORourke}. 
The long bottom line segment $T$
is visible from the upper half iff the observer can
peep through two successive gaps simultaneously,
i.e. from any position on the $\Theta(n^2)$ stripes
but not from the ellipses.
When moving from an ellipse to another, $T$
flashes into sight and is then hidden
again. There are $\Theta(n^4)$ such ellipses.
It is easy to see that this example is combinatorially
stable under small perturbation and hence can be made
non-degenerate.
\item[b)]
Consider Figure~\ref{f:linearVSP}:
Within each cell $C$ of the segment arrangement,
all segments are always visible;
hence $C$ is also a cell of the VSP.
And the exterior of $\calS$
gives rise to another 9 VSP cells.
\item[c)]
Recall the proof of Lemma~\ref{l:VSP}a);
but throw in the observation 
that crossing the line $L_{\vec a,\vec b}$
induced by two endpoints $\vec a$ and $\vec b$
(say, of segments $S_1,S_2\in\calS$, respectively)
does not induce a change in visibility
if $\vec a$ and $\vec b$ are occluded from each
other by some further segment $S\in\calS$:
compare Figure~\ref{f:VSPlines}.
Hence it suffices to consider at most as many
lines as the the number $m$ of edges in the Visibility Graph;
and these induce an arrangement of at most
quadratic complexity $\calO(m^2)$.
\item[d)]
Determine, according to \cite{VisibilityGraph}
in time $\calO(n\cdot\log n+m)$---or maybe
more practically in time $\calO(m\cdot\log n)$ \cite{OvermarsWelzl}---the 
$m$ lines mentioned in Item~c); then proceed as in Lemma~\ref{l:VSP}c).
\qed
\end{enumerate}
\end{proof}

\begin{figure}[htb]
\includegraphics[width=0.99\textwidth]{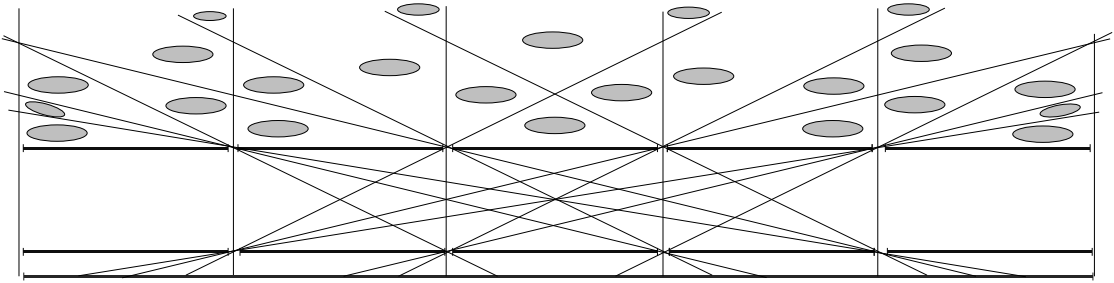}
\caption{\label{f:ORourke}Example of a segment
arrangement whose Visibility Space Partition
requires memory of order $n^4$}
\end{figure}

\begin{figure}[htb]
\begin{minipage}{0.45\textwidth}
\includegraphics[width=\textwidth,height=0.9\textwidth]{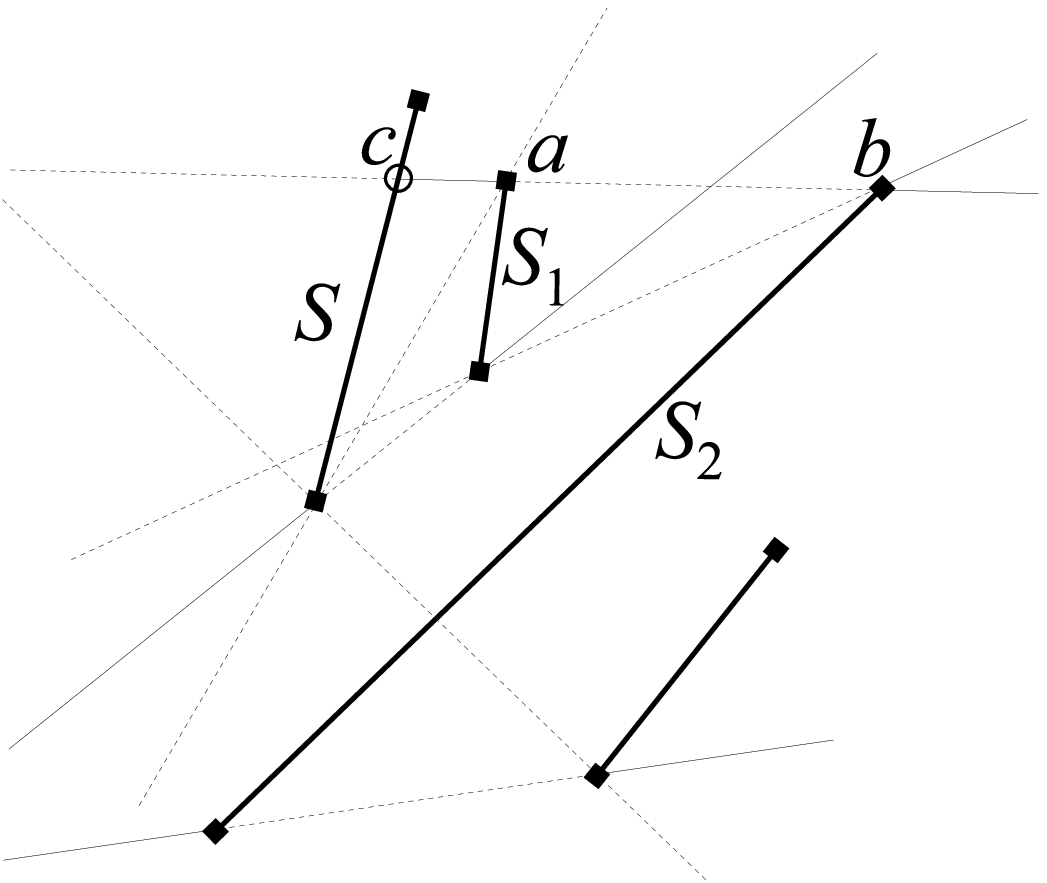}
\caption{\label{f:VSPlines}Visibility does \emph{not}
change to an observer crossing the \emph{dotted} 
(parts of) 
lines through pairs of segment endpoints.}
\end{minipage}\hfill\begin{minipage}{0.5\textwidth}
\includegraphics[width=\textwidth,height=0.84\textwidth]{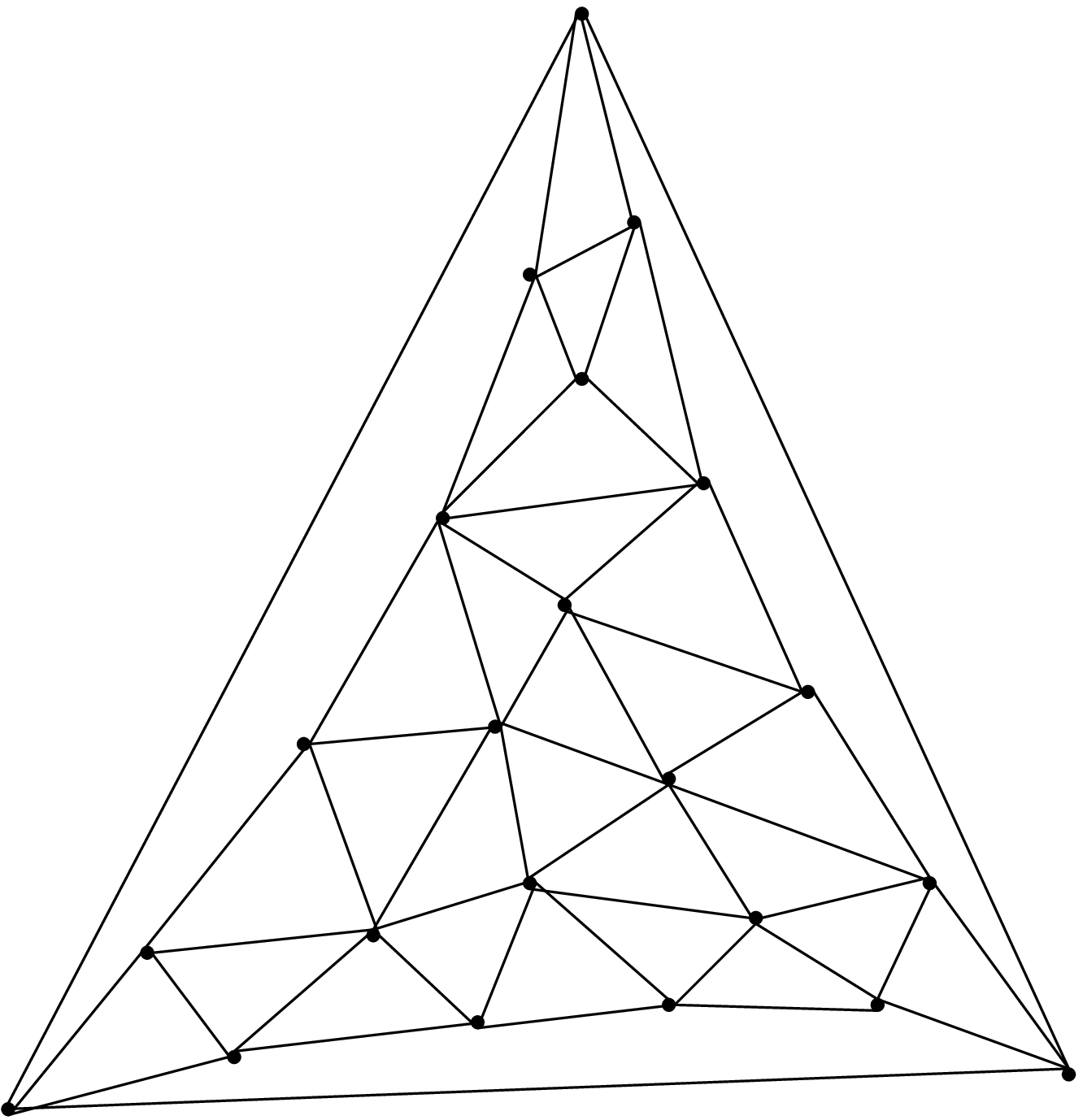}
\caption{\label{f:linearVSP}Example of a segment
arrangement whose Visibility Space Partition
has only linear size.}
\end{minipage}
\end{figure}

\subsection{Output-Sensitive VSP Calculation} \label{s:OutputSensitive}
In view of the large variation of VSP sizes from order $n$ to order $n^4$
according to Lemma~\ref{l:VSPsize}, the algorithms
indicated in Lemma~\ref{l:VSP}b+c)
for their calculation are reasonable only in case of large VSPs.
We now present an output-sensitive improvement
of Lemma~\ref{l:VSPsize}d):

\begin{theorem} \label{t:OutputSensitive}
In 2D, the data structure of Lemma~\ref{l:VSP}a-c) 
can be computed in time $\calO(n^2\cdot\log n+N\cdot\log n)$
in the sense of (the computational model referred to in)
Lemma~\ref{l:VSP}c),
where $N$ denotes the combinatorial complexity 
of $\VSP(\calS,\calS)$.
\end{theorem}
\begin{proof}
We start as in the proof Lemma~\ref{l:VSP}
with the order $n^2$ lines induced by all pairs
of segment endpoints. Now the idea is to 
extend Lemma~\ref{l:VSPsize}c), namely to take
into consideration only those \emph{parts} the lines 
lines are cut into,
which to cross actually changes the visibility.
Indeed, these sub-lines constitute the boundaries
of the cells of the VSP and therefore determine
its complexity.
\begin{enumerate}
\item[i)] Take such a line $L_{\vec a,\vec b}$
  passing through segment endpoints $\vec a$ and $\vec b$
  of segments $S_1$ and $S_2\in\calS$.
  To an observer crossing $L_{\vec a,\vec b}$,
  the visibility can, but need not, change---and
  we want to determine if and where it does.
  First observe that the middle part $(\vec a,\vec b)$ 
  of $L_{\vec a,\vec b}$ can be disposed off right away
  (unless $\vec a$ and $\vec b$ are endpoints of the same segment,
  but these $\calO(n)$ cases give rise to only $\calO(n^2)$
  combinatorial complexity anyway)
  because crossing it never changes the visibility;
  compare Figure~\ref{f:VSPlines}. \\
  Now consider the 
  two remaining unbounded rays of $L_{\vec a,\vec b}$,
  $L_{\vec a}$ starting from $\vec a$ 
  and $L_{\vec b}$ starting from $\vec b$.
  If, say, $L_{\vec a}$ intersects some other segment $S\in\calS$
  in some point $\vec c$,
  then traversing the part of $L_{\vec a}$ beyond that point
  does not affect the visibility either, as $S_2$ is 
  `shielded' from sight by $S$ anyway; again cf. Figure~\ref{f:VSPlines}.
  So let $L'_{\vec a}:=(\vec a,\vec c)$ in this case,
  $L'_{\vec a}:=L_{\vec a}$ otherwise; and similarly for $L'_{\vec b}$.
\\
  Now crossing $L'_{\vec b}$ at some point
  may or may not alter the visibility of (at least one of) $S_1,S_2,S$
  (the latter being a segment `\emph{opposite}' 
  to $\vec b$ along $L_{\vec a,\vec b}$)
  but if it does so, then it does so at \emph{every} point of $L'_{\vec b}$.
  Hence we will either keep the whole $L'_{\vec b}$, or drop it 
  entirely; similarly for $L'_{\vec a}$.
\item[ii)]
  Now since those two alternatives---namely keeping or dropping
  $L'_{\vec b}$---depend only on $S_1,S_2,S$,
  they can be distinguished in constant time.
  Moreover, after $\calO(n^2)$ preprocessing time and space for $\calS$,
  each $L_{\vec a,\vec b}$ can be decomposed into the two parts
  $L'_{\vec a}$ and $L'_{\vec b}$ as the result of a 
  \textsf{ray shooting} query among $\calS$ in logarithmic time;
  see e.g. \mycite{Theorem~3.2}{Flatland2}.\\
  The line parts $L'$ kept will in general intersect each other.
  So next cut them into non-intersecting maximal sub-segments.
  By the above observations, these constitute the boundaries of
  the VSP. And as a standard \textsf{segment intersection} problem,
  they can be determined in time $\calO(n^2\cdot\log n+N)$;
  cf. e.g. \mycite{Section~2.5}{deBerg}. \\
  The resulting (sub-)segments give rise
  to a planar subdivision. For instance they cannot contain
  leaf (e.g. degree-1) vertices: circling around such a
  vertex one way would change the visibility 
  and the other way would not.
  Therefore the data structure admitting logarithmic-time 
  point-location in the VSP can be calculated 
  in space $\calO(N)$ and time $\calO(N\cdot\log n)$,
  recall \mycite{Theorem~6.8}{deBerg}.
\item[iii)] 
  Determining the visibilities as in the proof of
  Lemma~\ref{l:VSP}b) yields a factor $n$ overhead;
  and the divide-and-conquer approach of Lemma~\ref{l:VSP}c) 
  seems inapplicable because of the correlations between
  segments in Step~ii), namely cutting off $L_{\vec a,\vec b}$
  induced by $S_1,S_2$ at the first further segment $S$ hit.
  On the other hand, each $L'_{\vec a}$ (and similarly for $L'_{\vec b}$) 
  by construction induces a definite change in visibility
  when crossed: we may presume this information to have
  been stored with $L'_{\vec a}$ at the beginning of Step~ii).
  Hence we may start at one arbitrary cell of the arrangement,
  calculate its visibility according to Lemma~\ref{l:IntervalTree},
  and then traverse the rest of the arrangement cell by cell
  while keeping track of the visibility changes induced by 
  (and stored with) each cell boundary.
\qed\end{enumerate}\end{proof}

\subsection{Visibility of One Single Target: Trading Time for Space} \label{s:Tradeoff}
The query time obtained in Lemma~\ref{l:VSP} is very fast:
logarithmic (i.e. optimally)
where, according to Paradigm~\ref{p:Paradigm}, sublinear suffices.
Quite intuitively it should be possible to reduce the memory
consumption at the expense of increasing the time bound.
We achieve this for the case of one target,
that is the decision version of visibility 
$\vec x\mapsto\Vis(\calS,\vec x,\{T\})\in\{0,1\}$:

\begin{theorem} \label{t:TimeVersusSpace}
For each $1\leq\ell\leq n$, 
Problem~\ref{p:Visibility} can be solved,
after $\calO(n^4\cdot\log^2n/\ell)$ time 
and space $\calO(n^4/\ell)$ preprocessing,
within query time $\calO(\ell\cdot\log n)$.
\end{theorem}
Such a trade-off result from time to space has become famous in the general
context of structural complexity \cite{TimeVersusSpace}.
Note that, obeying sublinear time, 
we can get arbitrarily close to cubic space---yet 
remain far from the joint resources consumption
of an interval tree (Section~\ref{s:IntervalTree}).
But first comes the already announced
\begin{lemma} \label{l:TimeVersusSpace}
Fix a collection $\calS\uplus\{T\}$ of $n+1$ non-crossing
segments in the plane.
Let $L_1,\ldots,L_k$ denote the $k=\binom{2n+2}{2}$ lines
induced by the pairs of endpoints of segments in $\calS\cup\{T\}$.
For an observer moving in the plane, the weak visibility
of $T$ can change only as she crosses 
\begin{itemize}
\item either one of the lines $L_i$ intersecting $T$ 
\item or someline supporting a segment $S\in\calS$;
\end{itemize}
compare Figure~\ref{f:tradeoff}.
\end{lemma}
\begin{proof}
Standard continuity argument:
Let $\vec p$ denote the observer's position
and suppose point $\vec x\in T$ is visible,
i.e. the segment $[\vec p,\vec x]$ does not intersect
$S\in\calS$. Now move $\vec p$ until $\vec x$
is just about to become hidden behind $S\in\calS$.
Then start moving $\vec x$ on $T$ such as to remain visible.
Keep moving $\vec p$ and adjusting $\vec x$:
this is possible (at least) as long as the line
through $\vec p$ and $\vec x$ avoids 
all endpoints of $\calS\cup\{T\}$.
\qed\end{proof}

\begin{figure}[htb]\centerline{%
\includegraphics[width=0.9\textwidth]{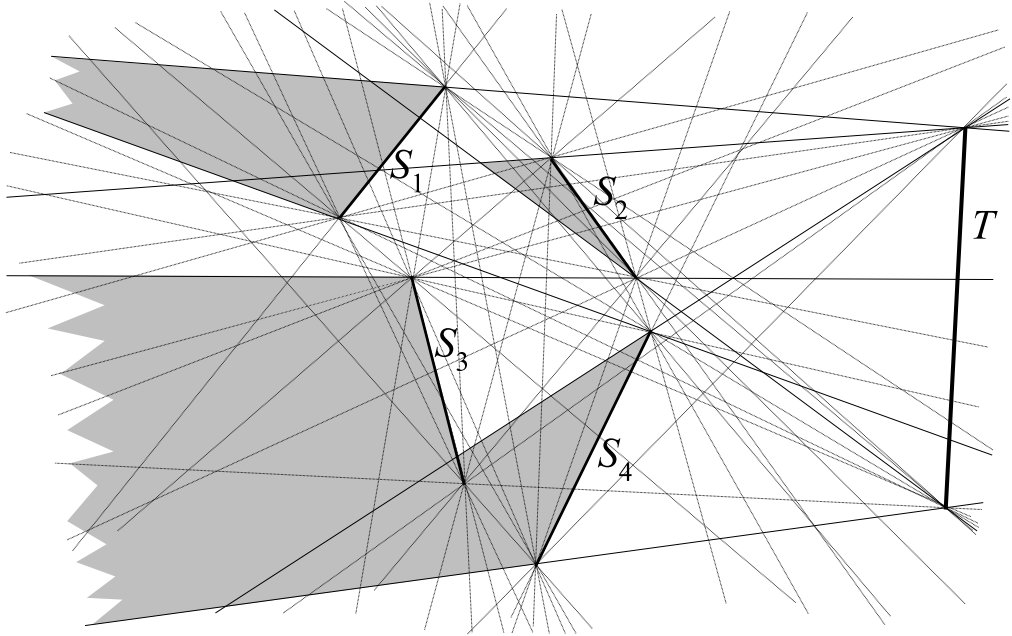}}%
\caption{\label{f:tradeoff}Regions where segment $T$ is weakly visible
through $\calS=\{S_1,S_2,S_3,S_4\}$ are delimited by lines through
endpoints of $\calS\cup\{T\}$ intersecting $T$
and by the segments themselves.}
\end{figure}

\begin{proof}[Theorem~\ref{t:TimeVersusSpace}]
Consider, as in the proof of Lemma~\ref{l:VSP}, 
the $\calO(n^2)$ lines induced by pairs of segment endpoints of $\calS$.
Consider the intersections of these lines with $T$ (if any).
Partition $T$ into $\calO(\ell)$ sub-segments $T_1,\ldots,T_\ell$, 
each intersecting $\calO(n^2/\ell)$ of the above lines.
For each piece $T_i$, take the arrangement $\calA_i$ of size
$\calO\big((n^2/\ell+n)^2\big)$ 
induced by those lines intersecting $T_i$,
and all $\calO(n)$ lines through one endpoint of $T_i$ and one of some $S\in\calS$,
and all $\calO(n)$ lines supporting segments from $\calS$.
By Lemma~\ref{l:TimeVersusSpace}, within each cell $C$ of $\calA_i$,
the weak visibility of $T_i$ is constant (either \texttt{yes}
or \texttt{no}) and can be stored with $C$: 
Doing so for each $\calA_i$ ($1\leq i\leq \ell\leq n$)
and each of the $\calO(n^4/\ell^2+n^3/\ell+n^2)$ cells $C$ of $\calA_i$
uses memory of order $\calO(n^4/\ell+n^3+n^2\ell)=\calO(n^4/\ell)$ 
as claimed; and corresponding time according to Lemma~\ref{l:VSP}c).

Then, given a query point $\vec p\in\IR^2$,
locating $\vec p$ in each arrangement $\calA_i$
takes total time $\calO(\ell\cdot\log n)$;
and yields the answer to whether $T_i$ is 
weakly visible from $\vec p$ or not.
Now $T$ itself is of course visible iff some $T_i$ is:
a disjunction computable in another $\calO(\ell)$ steps.
\qed\end{proof}
We even can combine Lemma~\ref{l:VSPsize}d) with
Theorem~\ref{t:TimeVersusSpace} to obtain 

\begin{scholium}
\label{s:TimeVersusSpace}
For each $1\leq\ell\leq n$, 
Problem~\ref{p:Visibility} can be solved,
after preprocessing $\calS$ 
in $\calO(n\cdot\log n+m^2\cdot\log^2n/\ell)$ time
into an $\calO(m^2/\ell)$ size data structure,
within query time $\calO(\ell\cdot\log n)$
where $m$ denotes the size (number of edges)
of the Visibility Graph of $\calS$.
\end{scholium}
\begin{proof}
Instead of considering, and partitioning into $\ell$ groups, 
all $\calO(n^2)$ lines induced by pairs of segment endpoints,
do so only for the $\calO(m)$ lines induced by pairs
segment endpoints visible to each other.
\qed\end{proof}

\section{Approximate Visibility Counting} \label{s:Approximate}
Lacking deterministic exact algorithms 
for calculating visibility counts satisfying both
time and space requirements, we now
resort to approximations: 
of $\Vis(\calS,\vec x,\calS)$ 
up to prescribable absolute error $k\in\IN$
or, equivalently, of the visibility \emph{ratio}
$\Vis(\calS,\vec x,\calS)/\Card(\calS)$ up to
absolute error $\epsilon=k/\Card(\calS)$;
recall Hypothesis~\ref{h:Count}.

\begin{remark}
Relative errors make no sense as there is always a
viewpoint $\vec x$ with $\Vis(\calS,\vec x,\calS)=1$.
\end{remark}
Corollary~\ref{c:Tradeoff}, the main result of this section,
presents a randomized approximation 
within sublinear time using almost cubic space in the worst-case
and almost linear space in the `typical' one.

\subsection{Deterministic Approach: Relaxed VSPs} \label{s:Coarse}
Visibility space partitions, and the algorithms based upon them,
are so memory expensive because they discriminate
(i.e. introduces separate arrangement cells for) 
observer positions 
whose visibility differs by as little as one;
recall Definition~\ref{d:VSP}.
It seems that considerably more (time and) space efficient 
algorithms may be feasible by partitioning 
observer space into (or merely covering it by) more coarse classes:

\begin{definition} \label{d:Coarse}
Fix $k\in\IN$ and collections $\calS$ and $\calT$
of non-intersecting segments in the plane.
Some covering $\{C_1,\ldots,C_I\}$ of $\IR^2$ is
called a \textsf{$k$-relaxed VSP of $(\calS,\calT)$} if
\[ \forall 1\leq i\leq I \;
\forall \vec p,\vec q\in C_i: \quad
 \Vis(\calS,\vec p,\calT)-\Vis(\calS,\vec q,\calT)\;\leq\;k \enspace . \]
In the sequel we shall restrict to 
$k$-relaxed VSPs which constitute planar subdivisions
(i.e. each $C_i$ being a simple polygon);
and refer to their \textsf{size} in the sense of Definition~\ref{d:VSP}.
\end{definition}
Indeed, such VSPs 
allow for locating a given observer position $\vec x$ in logarithmic time
to yield a cell $C_i\ni\vec x$ which, during preprocessing,
had been assigned a value $\Vis(\calS,\vec q,\calT)$
approximating $\Vis(\calS,\vec x,\calT)$ up to absolute
error at most $k$.

\begin{example}
For $\Card\calT\leq k$, the trivial planar subdivision $\{\IR^2\}$
is a $k$-relaxed VSP of $(\calS,\calT)$. \\
In particular the quartic lower size bound 
of Lemma~\ref{l:VSPsize}b) applies only 
to 0-relaxed VSPs but breaks down for $k\geq1$.
\end{example}
This example suggests that much smaller 
(e.g. worst-case quadratic) sizes might
become feasible when considering $k$-relaxed VSPs
for, say, $k\approx\sqrt{n}$
or even $k\approx n/\log n$.
Indeed we have the following lower and upper bounds:
\begin{proposition} \label{p:Coarse}
\begin{enumerate}
\item[a)]
For each $n,k$ there exists a non-degenerate family
$\calS$ of segments in the plane such that
any $k$-relaxed VSP has size at least 
$n\cdot\lfloor(n-1)/(k+1)\rfloor$.
\item[b)]
There also exist such families such that
any $k$-relaxed VSP has size at least
$\Omega(n^4/k^4)$.
\item[c)]
Let $\calS$ be a non-degenerate family of $n$ segments in the
plane and $N$ the size of its VSP.
Then there exists a $k$-relaxed VSP of size $\lfloor N/(k+1)\rfloor$.
\item[d)]
There also exists a $k$-relaxed VSP 
of size $\calO(m^2/k^2)$,
where $m$ denotes the size of the Visibility Graph of $\calS$.
\item[e)]
In fixed dimension $d$, (Definition~\ref{d:Coarse} and) 
Item~c) generalizes to $k$-relaxed VSPs
of size $\calO(n^{d^2}\!\!/k^d)$.
\end{enumerate}
\end{proposition}
Recall that $N\leq m^2\leq n^4$, 
thus leaving quadratic gap between a) and d)
for $k$ small; and between b) and d) for $k$ large.
Item~c) succeeds over d) in cases where
$N$ asymptotically does not exceed $m^2/k$.
\begin{figure}[htb]\centerline{%
\includegraphics[width=0.9\textwidth]{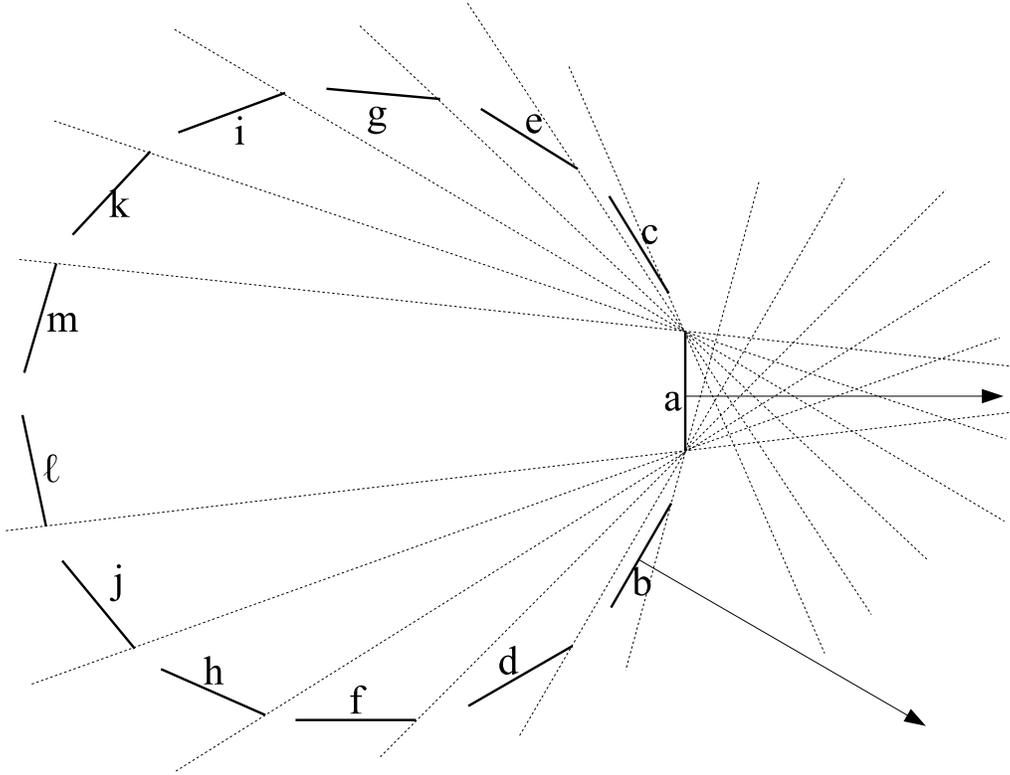}}%
\caption{\label{f:Coarse}Crossing each dashed line 
increases visibility by one segment; hence keeping
only every $k$-th leads to visibility count
variations within a cell of at least $k$.}
\end{figure}
\begin{figure}[htb]\centerline{%
\includegraphics[width=0.99\textwidth]{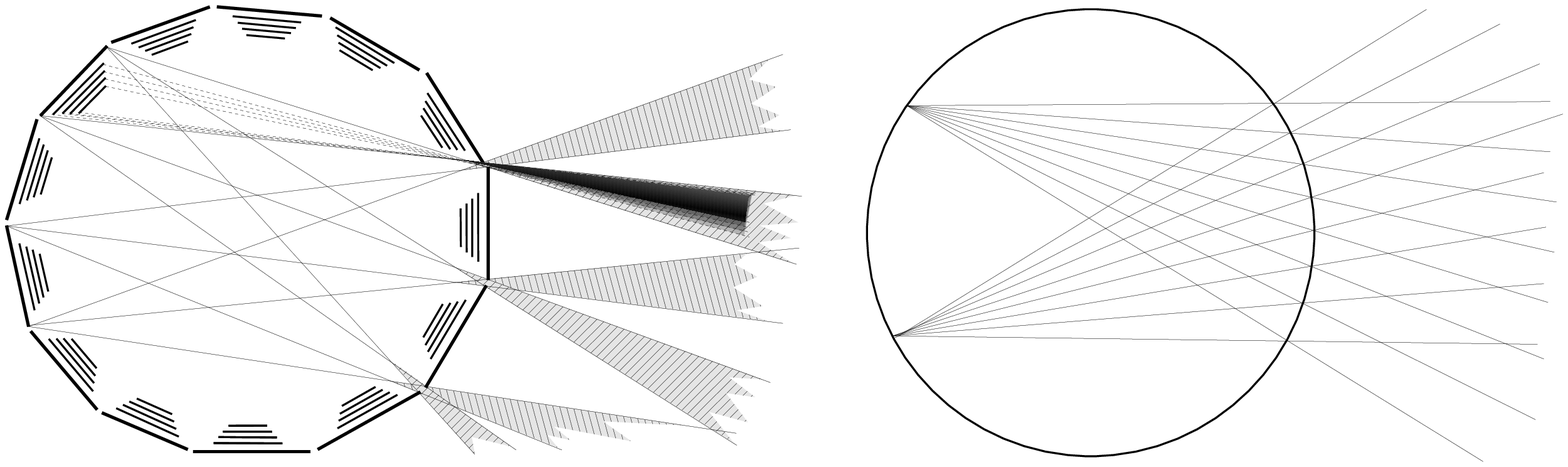}}%
\caption{\label{f:Coarse2}Illustrating the proof
of Proposition~\ref{p:Coarse}b).}
\end{figure}
\begin{proof}
\begin{enumerate}
\item[a)]
Consider Figure~\ref{f:Coarse} with $n=13$ segments
which obviously generalizes to arbitrary $n$.
Moving from segment~$a$ along the arrow, each time
crossing a dashed line amounts to an increase in
visibility from $\{a\}$ via $\{a,b\}$, $\{a,b,c\}$
and so on up to entire $\calS$.
Hence to obtain cells of viewpoints with visibility
varying by at most $k$, we must keep at least
every $(k+1)$-st dashed border, that is 
$\lfloor(n-1)/(k+1)\rfloor$ out of $n-1$.
By symmetry, the same argument applies when moving
from segment~$b$ along the arrow, or from any other
segment.
\item[b)]
A closer look at Figure~\ref{f:Coarse} reveals it
to induce a VSP of size $\Omega(n^4)$: For any segment,
there are linearly many separate cones from which it can
be seen through the gaps between the other segments;
hence we have quadratically many cones, 
of which almost any two intersect;
compare Figure~\ref{f:Coarse2}, particularly its right part.
\\
Now in order to argue about $k$-relaxed VSPs,
replace in Figure~\ref{f:Coarse} each single segment by
$k$ scaled and shifted
copies as indicated in the left of Figure~\ref{f:Coarse2};
that is we now have a scene of size $N=n\cdot k$.
Observe that, when entering and passing through
a cone of visibility, the number of segments visible
increases from its original value to an additional $k$
(drawn in levels of gray). The visibility number thus
varies by $k+1$, requiring a $k$-relaxed VSP to
subdivide the cone; indeed the entire cone!
By the above considerations, these necessary boundaries
induce an arrangement of complexity $\Omega(n^4)$
which, expressed in the size $N$ of the scene, 
is $\Omega(N^4/k^4)$ as claimed.
\item[c)]
Classify the cells of $\VSP(\calS,\calS)$ according
to their visibility count; and let $N_i$ denote 
the number of boundary segments of $\VSP(\calS,\calS)$ 
separating a cell with visibility count $i$ from one with
visibility count $i+1$. Since any boundary segment
does so for some $i=1,\ldots,n-1$,
\begin{eqnarray*} 
N&=&N_1+N_2+\cdots+N_{n-1} \\
&=&\phantom{+}\;(N_1+N_{k+2}+N_{2(k+1)+1}+\ldots)
+(N_2+N_{k+3}+N_{2(k+1)+2}+\ldots)\\
&&
+\;(N_3+N_{k+4}+N_{2(k+1)+4}+\ldots)
+\cdots+
(N_{k+1}+N_{2(k+1)}+N_{3(k+1)}+\ldots) \enspace .
\end{eqnarray*}
Hence by pigeonhole principle there exists
$1\leq\kappa\leq k+1$ such that
$\lfloor N/(k+1)\rfloor\geq
N_{\kappa}+N_{k+1+\kappa}+N_{2(k+1)+\kappa}+\ldots$
Now keep from $\VSP(\calS,\calS)$ exactly those
boundary segments that either separate cells
with visibility count $\kappa$ from ones with
visibility count $\kappa+1$, or cells with
visibility count $k+1+\kappa$ from ones
with visibility count $k+2+\kappa$, or
cells with visibility count $2(k+1)+\kappa$
from visibility count $2(k+1)+\kappa+1$ 
and so on.
\\
By the above considerations, this planar subdivision
has complexity at most $\lfloor N/(k+1)\rfloor$.
And by construction, it joins cells having 
visibility counts $j\cdot(k+1)+\kappa+1$,
$j\cdot(k+1)+\kappa+2$, $j\cdot(k+1)+\kappa+3$, 
\ldots, $j\cdot(k+1)+(k+1)+\kappa$;
but maintains their separation from cells with
visibility count $(j+1)\cdot(k+1)+\kappa+1$,
because those boundary segments
are precisely the ones deliberately kept.
Hence the joined coarsened (i.e. super-) cells indeed contain
only viewpoints with visibility count differing
by at most $k$.
\item[d)]
Recall the collection
$\calL$ of $m$ lines used in the proof of Lemma~\ref{l:VSPsize}c).
Now let $r:=m/(k-1)$, $k>1$, and apply the \textsf{Cutting Lemma} of
Chazelle, Friedman, and Matou\v{s}ek \mycite{Lemma~4.5.3}{Matousek}:
\begin{quote}
There exists a subdivision of the plane into
$\calO(r^2)=\calO(m^2/k^2)$ generalized (i.e. not necessarily closed)
triangles $\Delta_i$ such that the interior of each $\Delta_i$
is intersected by at most $m/r=k-1$ lines from $\calL$.
\end{quote}
Since visibility changes occur only when crossing
lines in $\calL$, the latter means that the visibility
counts within each $\Delta_i$ differ by at most $k$.
\item[e)]
Instead of $\calO(r^2)$ generalized triangles as in d),
now employ simplicial cuttings of size $\calO(r^d)$
according to \mycite{Theorem~6.5.3}{Matousek}.
\qed\end{enumerate}\end{proof}
We remark that Item~c) considers a certain
sub-arrangement of the 0-relaxed VSP,
whereas the planar subdivision due to Item~d)
uses cell boundaries \emph{not} necessarily belonging 
to the VSP. Also, the size $\calO(r^d)$ of the
cuttings employed in Items~d) and e) is known to
be optimal in general; but this optimality does
not necessarily carry over to our application
in visibility, recall the gaps between lower
and upper bounds in Items~a) to d).

Substituting $k=\delta\cdot n$ yields 

\begin{corollary} \label{c:Coarse}
For each collection $\calS$ of non-degenerate segments in the plane
and $k\in\IN$, there exists a data structure of size 
$\calO\big(\min\{N/(\delta\cdot n),m^2/(\delta\cdot n)^2\}\big)$ 
that allows to approximate, 
given $\vec x\in\IR^2$,
the visibility ratio $\Vis(\calS,\vec x,\calS)/\Card(\calS)$ 
up to absolute error $\delta>0$
in time $\calO(\log n)$.\\
\end{corollary}
Notice that we only claim the \emph{existence} of
such small data structures. In order to \emph{construct}
them, the proofs of Proposition~\ref{p:Coarse}c) and d)
both proceed by first calculating the 0-relaxed VSP
and then coarsening it.
Specifically for Proposition~\ref{p:Coarse}d),
a \textsf{Triangular Cutting} 
can be obtained in time $\calO(m\cdot r\cdot\polylog)$ \cite{Agarwal};
but determining the visibility count for each triangle $\Delta_i$
costs $\calO(n\cdot\log n)$ according to Lemma~\ref{l:RotSweep};
or can be taken from the 0-relaxed VSP.
The preprocessing time for Proposition~\ref{p:Coarse}c+d) 
thus is, up to polylogarithmic factors, 
that of Theorem~\ref{t:OutputSensitive},
i.e. roughly $\calO(N)$:
independent of, and not taking advantage of large values of, $k$.
For the configuration from Lemma~\ref{l:VSPsize}a) for instance
(recall Figure~\ref{f:ORourke}), this results in a preprocessing time
of order $n^4$ although the resulting $1$-relaxed VSP has only size $\calO(n^2)$.
Alternatively, apply Lemma~\ref{l:RotSweep} to (one point from)
each triangle $\Delta_i$ to obtain a running time of roughly output$\times n$, 
that is still off optimal by one order of magnitude.
And finally, the asymptotically `small' size and time
for calculating triangular cuttings hide in the
big-Oh notation some large constants
which are believed to prevent practical applicability.

\subsection{Random Sampling} \label{s:Chernoff}
Both size and query time of the data structure due to Corollary~\ref{c:Coarse}
are rather low; but because of the infavourable
preprocessing time and hidden big-Oh overhead
indicated above, 
we now proceed to random sampling,
based on a rather simple generic algorithm:

\begin{algorithm} \label{a:Sampling} ~
\begin{enumerate}
\item[i)] Guess a sample target $\calT\subseteq\calS$
of size $m$.
\item[ii)]
Calculate the count $\Vis(\calS,\vec x,\calT)$
of objects in $\calT$ visible through $\calS$.
\item[iii)]
Return the ratio $\Vis(\calS,\vec x,\calT)/\Card(\calT)$;
\item[iv)]
and hope that it does not deviate too much
from the `true' value $\Vis(\calS,\vec x,\calS)/\Card(\calS)$.
\end{enumerate}
\end{algorithm}
Item~iv) is justified by the following

\begin{lemma} \label{l:Chernoff}
Fix $\vec x\in\IR^d$ and $\delta>0$,
then choose $\calT\subseteq\calS$
as $m$ independent identically distributed 
random draws from $\calS$. It holds
\[ \Prob_{\calT}\Big[\big|\Vis(\calS,\vec x,\calT)/m\,-\,
\Vis(\calS,\vec x,\calS)/n\big|\:\geq\:\delta\Big]
\;\leq\; 2\cdot e^{-2m\cdot\delta^2} \]
\end{lemma}
In other words: In Algorithm~\ref{a:Sampling}
taking $m$ (quadratic in the aimed \emph{absolute} accuracy
$\delta$ but) \emph{constant} with respect to the
scene size $n$ suffices to achieve the desired 
approximation with constant probability; slightly
increasing it further 
amplifies exponentially the chance for success.

\begin{remark}
\begin{enumerate}
\item[a)]
It is easy to see that a fixed \emph{relative} accuracy
can be attained, for $\Vis(\calS,\vec x,\calS)/n\to0$,
only by samples of size $m\to n$: If only one segment
is visible, it must get sampled to be detected.
\item[b)]
Also the visibility of the sample is crucially to
be considered with respect to the \emph{entire} scene,
i.e. $\Vis(\calS,\vec x,\calT)$ rather than
$\Vis(\calT,\vec x,\calT)$.
\end{enumerate}
\end{remark}

\begin{proof}[Lemma~\ref{l:Chernoff}]
It is well-known \cite{Motwani,AlonSpencer}
that a sum $X:=\sum_{i=1}^m X_i$
of independent $\{0,1\}$ trials $X_1,\ldots,X_m$
satisfies the \textsf{Chernoff--Hoeffding Bound}
$\Prob\big[|X/m-\mu|\geq\delta\big]\leq2\cdot\exp(-2m\cdot\delta^2)$
where $\mu$ denotes the expectation of $X_i$.
In our case, let $X_i$ denote the event that
the $i$-th draw $S_i\in\calS$ 
is visible from $\vec x$ through $\calS$.
This happens with probability $\mu=\Vis(\calS,\vec x,\calS)/n$,
hence $X=\Vis(\calS,\vec x,\calT)$.
\qed\end{proof}
%

\subsection{The VC-Dimension of Visibility} \label{s:VCdim}
Note that the random experiment $\calT$ 
and the probability analysis of its properties
in Lemma~\ref{l:Chernoff} 
holds for each $\vec x$ but not
uniformly in $\vec x$.
This means for our purpose to re-sample $\calT\subseteq\calS$ at every frame.
On the other hand, 
the above considerations have not exploited any geometry.
An important connection between combinatorial 
sampling and geometric properties is captured by the
\textsf{Vapnik--Chervonenkis Dimension} \cite{AlonSpencer}:

\begin{fact} \label{f:VCdim}
Let $X$ be a set and $\calR$ a collection of subsets $R\subseteq X$.
Denote by 
\begin{equation} \label{e:VCdim}
d\;:=\;\VCdim(X,\calR) \;\;:=\;\;
\max\big\{\Card Y\big| Y\subseteq X,
\{Y\cap R:R\in\calR\}=2^Y\big\}
\end{equation}
the VC-Dimension of $(X,\calR)$.
\begin{enumerate}
\item[a)]
For $Y\subseteq X$, $n=\Card(Y)$,
~ $\Card\{Y\cap R:R\in\calR\}\;\leq\;\sum_{i=0}^d\binom{n}{d}
\;\leq\; n^d$.
\item[b)]
Let $Y\subseteq X$ be random of 
$\Card(Y)\geq\max\big\{4/\delta\cdot\log\tfrac{2}{p},8d/\delta\cdot\log\tfrac{8d}{\delta}\big\}$.
Then with probability at least $1-p$, it holds
for each $R\in\calR$: ~
$\Card(X\cap R)\geq\delta\cdot\Card(X) \;\Rightarrow\; Y\cap R\not=\emptyset$.
\item[c)]
Let $Y\subseteq X$ be random of 
$\Card(Y)\geq\Omega\big((d\cdot\log\tfrac{d}{\delta}+\log\tfrac{1}{p})/\delta^2\big)$.
Then with probability at least $1-p$ it holds
for each $R\in\calR$: ~
$\big|\Card(X\cap R)/\Card(X)\,-\,\Card(Y\cap R)/\Card(Y)\big| 
\;\leq\;\delta$.
\end{enumerate}
\end{fact}

\begin{lemma} \label{l:VCdim}
Fix a collection $\calS$ of $n$ non-crossing
$(d-1)$--simplices in $\IR^d$. 
\begin{enumerate}
\item[a)]
Define $X:=\calS$ and $\calR:=\{\VIS(\calS,\vec x,\calS):\vec x\in\IR^d\}$.
Then $\VCdim(X,\calR)\leq d^2\cdot\big(\log n+\calO(1)\big)$.
\item[b)]
A random subset $\calT\subseteq\calS$ of cardinality
$m\geq\Omega\big(d^2\cdot\log(n)\cdot\log(d\cdot\log n/\delta)/\delta\big)$
satisfies with constant (and easily amplifiable) probability
that, whenever at least a $\delta$-fraction
of the simplices of $\calS$ are visible from $\vec x\in\IR^d$
through $\calS$, then so is some simplex of $\calT$.
\item[c)]
A random subset $\calT\subseteq\calS$ of cardinality
$m\geq\Omega\big(d^2\cdot\log(n)\cdot\log(d\cdot\log n/\delta)/\delta^2\big)$
satisfies with constant (yet easily amplifiable) probability
that $\Vis(\calS,\vec x,\calT)/\Card(\calT)$ deviates from
$\Vis(\calS,\vec x,\calS)/\Card(\calS)$ absolutely
by no more than $\delta$.
\item[d)]
The bound obtained in a) is asymptotically 
optimal with respect to $n$: In $\IR^2$
there exist non-degenerate collections $\calS=X$ of $n$ 
line segments such that $\VCdim(X,\calR)\geq\log n-\calO(\loglog n)$.
\end{enumerate}
\end{lemma}
\begin{proof}
\begin{enumerate}
\item[a)]
Lemma~\ref{l:VSP}d) implies $\Card\calR\leq\calO(n)^{d^2}$.
Hence, in Equation~(\ref{e:VCdim}),
$2^Y=\{Y\cap R:R\in\calR\}$ requires $2^{\Card Y}\leq\Card\calR$
and therefore $\Card Y\leq d^2\cdot\big(\log n+\calO(1)\big)$.
\item[b)] and c) follow by plugging Item~a) into Fact~\ref{f:VCdim}b+c).
\item[d)]
Figure~\ref{fig:VCdim} below obviously extends to
the construction of $k$ segments of which,
using $2^k\cdot k/2$ additional segments as `shields',
each subset appears as a visible set
$\VIS(\calS,\vec x,\calS)$ for some $\vec x$.
That is a scene $\calS$ of size $n=k+2^k\cdot k/2$
containing a subset $Y$ of size $k=\log n-\calO(\loglog n)$
as in Equation~(\ref{e:VCdim}).
\qed\end{enumerate}\end{proof}

\begin{figure}[htb]\centerline{%
\includegraphics[width=0.8\textwidth,height=0.25\textheight]{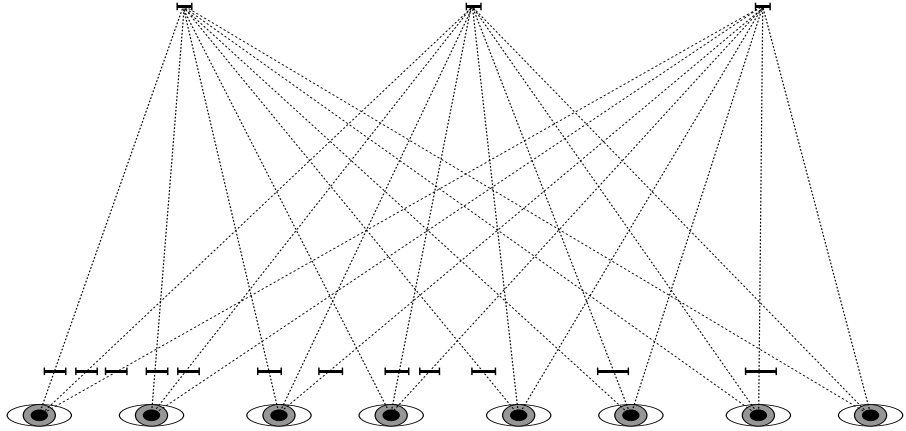}}
\caption{\label{fig:VCdim}3 segments of which,
shielded by 12 further segments, each
of its 8 subsets appears as visible set.}
\end{figure}
\subsection{Main Result} \label{s:MainTheory}
Lemma~\ref{l:VCdim}c) enhances Lemma~\ref{l:Chernoff}:
The latter is concerned with the probability of a \emph{constant}-size
sample $\calT$
to be representative (i.e. to approximate the visibility
ratio) with respect to a \emph{fixed} viewpoint---i.e. our
application would (have to) re-sample in each frame!
The former lemma on the other hand
asserts that a \emph{polylogarithmic}-size sample,
drawn once and for all, be suitable with respect to
\emph{all} viewpoints!
In particular we may preprocess the visibility
of each $T\in\calT$ separately according to 
Theorem~\ref{t:TimeVersusSpace} and obtain,
employing Scholium~\ref{s:TimeVersusSpace}:

\begin{corollary} \label{c:Tradeoff}
Given $0<\delta<1$, a collection $\calS$ of $n$ non-crossing
segments in the plane ($d=2$), and $1\leq\ell\leq n\leq m$
where $m$ denotes the size of the Visibility Graph of $\calS$.
Then a randomized algorithm can preprocess $\calS$ within time 
$\calO\big(n\cdot\log n+m^2\cdot\polylog n\cdot\log\tfrac{1}{\delta}/(\ell\cdot\delta^2)\big)$ 
and space 
$\calO\big(m^2\cdot\polylog n\cdot\log\tfrac{1}{\delta}/(\ell\cdot\delta^2)\big)$ 
into a data structure having with high probability the following property:
Given $\vec x\in\IR^2$, one can approximate 
the visibility ratio $\Vis(\calS,\vec x,\calS)/\Card(\calS)$ 
up to absolute error at most $\delta$
in time $\calO\big(\ell\cdot\polylog n\cdot\log\tfrac{1}{\delta}/\delta^2\big)$.
\end{corollary}
Again, note the trade-off between space and query time
gauged by the parameter $\ell$.
And, remembering the paragraph following Lemma~\ref{l:VSPsize},
$m$ is `typically' linear in $n$; hence choosing $\ell=n^{1-\epsilon}$,
the space can be
made arbitrarily close to linear while maintaining sublinear query time,
thus complying with Paradigm~\ref{p:Paradigm}!

\section{Empirical Evaluation} \label{s:Evaluation}
The present section demonstrates the practical applicability
of the algorithm underlying Corollary~\ref{c:Tradeoff}:
It is (not trivial but neither) too hard to implement,
constants hidden in big-Oh notation are modest,
and query time can indeed be traded for memory.

Measurements were obtained on an
\textsf{Intel\textregistered{} Dual\texttrademark2 CPU 6700}
running at 2.66GHz under \textsf{openSUSE~11.0}
equipped with 4GB of RAM. The implementation
is written in \textsf{Java} version 6 update 11.
Calculations on coordinates use exact rational
arithmetic based on \textsc{BigInteger}s.

\subsection{Benchmark Scenes}
We consider three kinds of `virtual scenes' in 2D,
that is collections of non-intersecting line segments,
compare Figure~\ref{f:scenes}:

\begin{figure}[htb]
\includegraphics[width=0.3\textwidth]{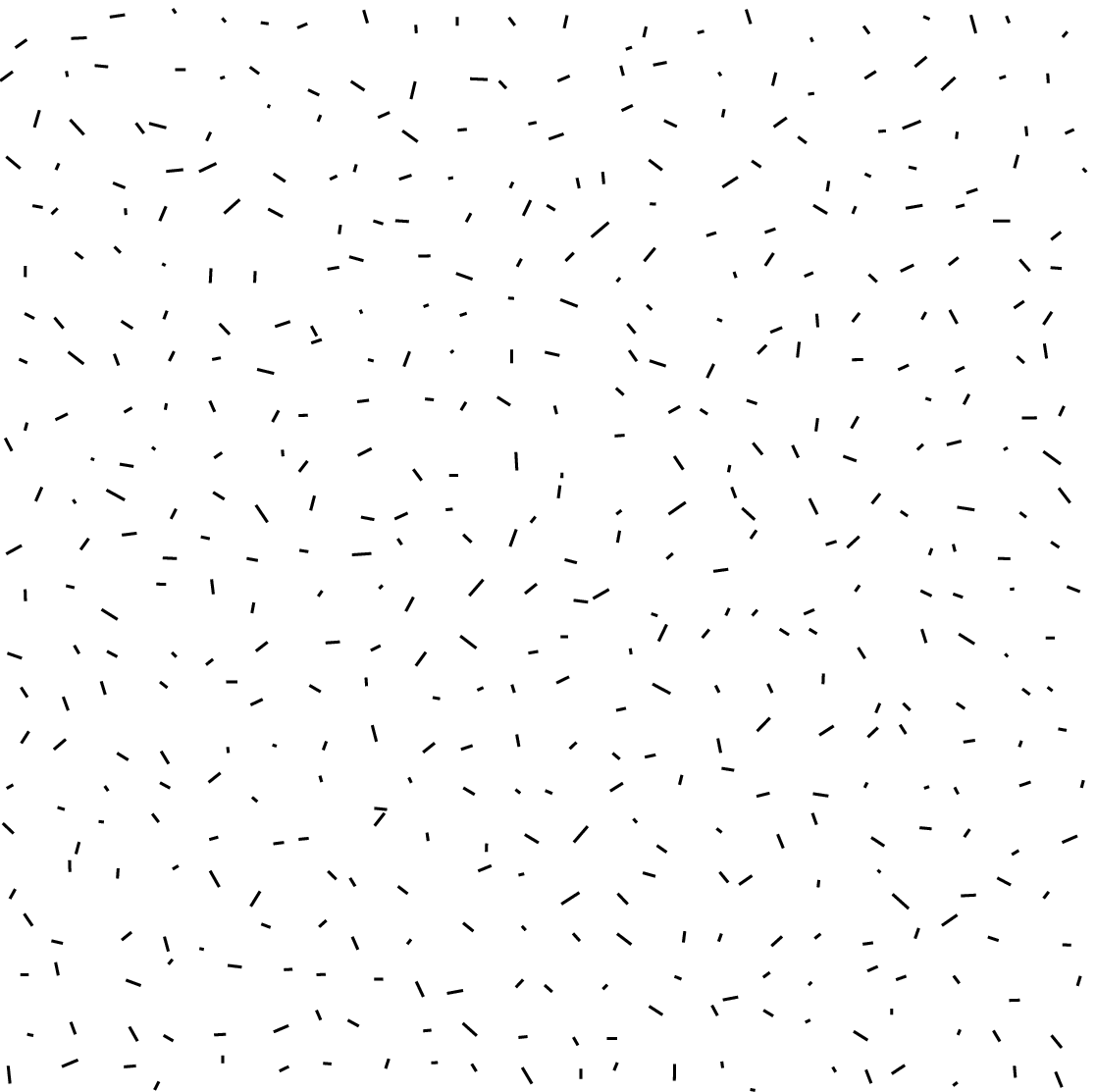}\hfill%
\includegraphics[width=0.3\textwidth]{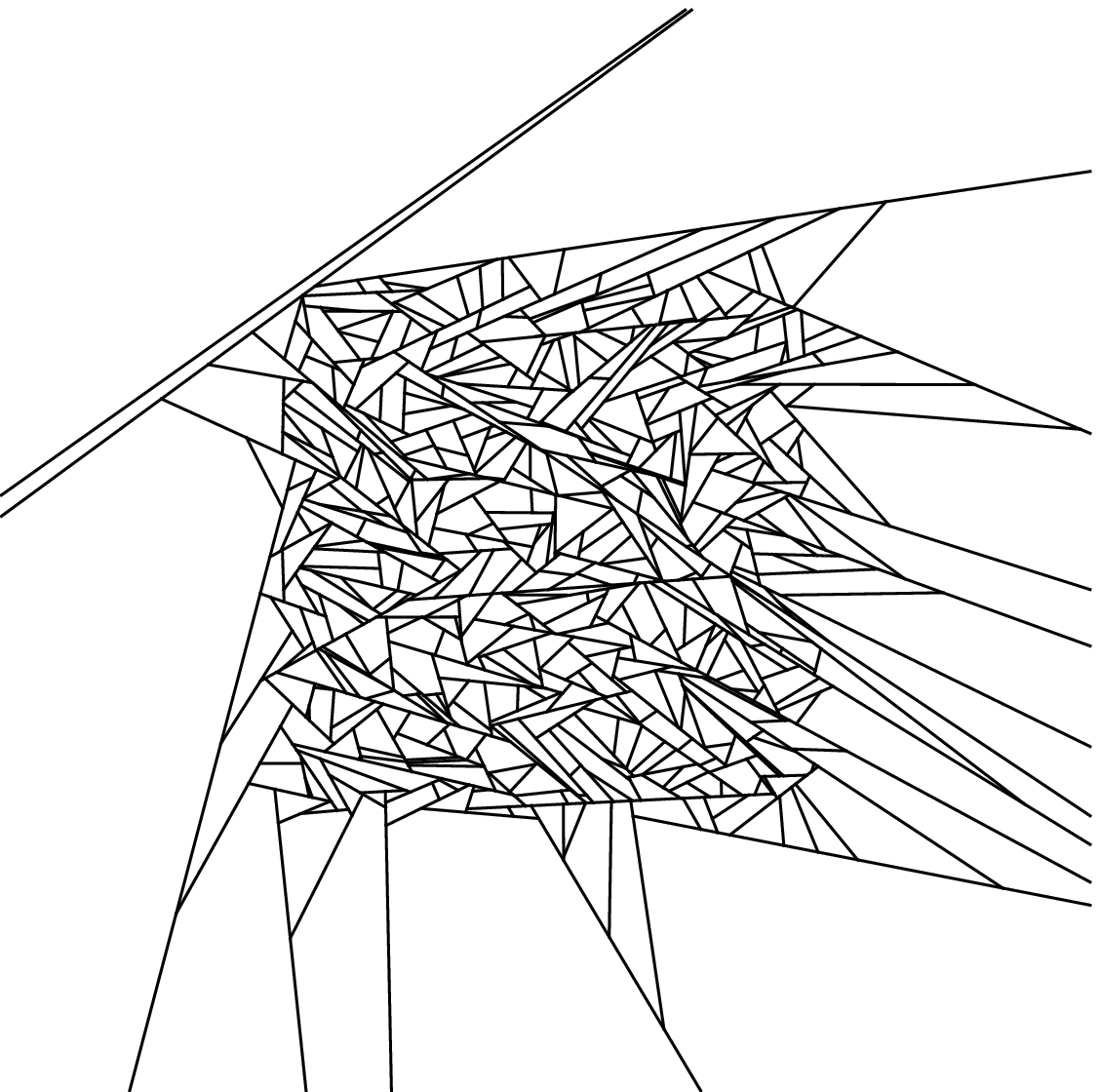}\hfill%
\includegraphics[width=0.3\textwidth]{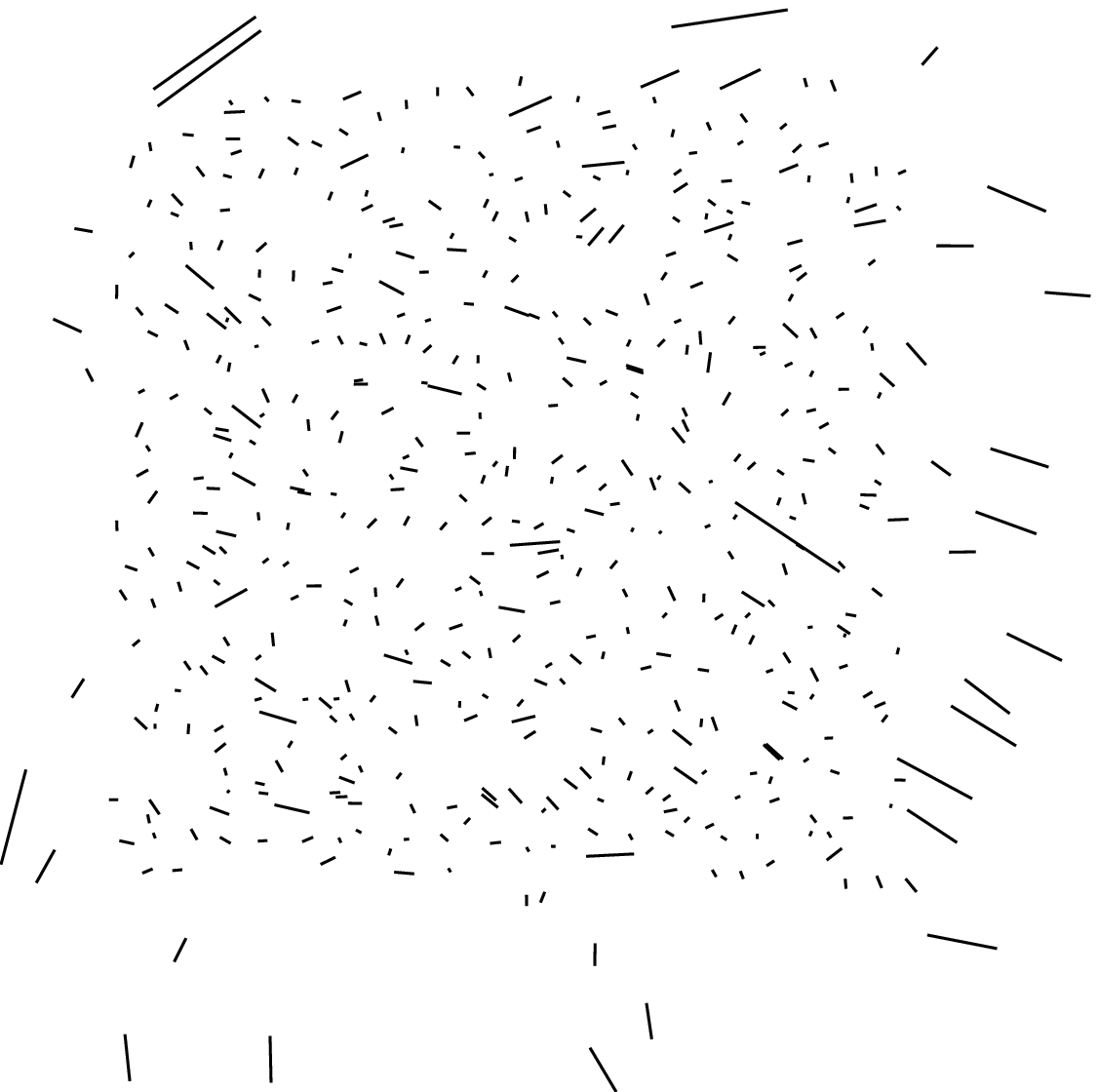}%
\caption{\label{f:scenes}Scene types A, B, and C.}
\end{figure}

\begin{enumerate}\itemsep0pt%
\item[A)] Sparse scenes representing 
  forest-like virtual environments with long-range visibility;
\item[B)] cellular scenes representing
  architectual virtual environments with visibility essentially
  limited to the room the observer is presently in;
\item[C)] and an intermediate of both.
\end{enumerate}
As indicated by the above classification,
these scenes contain some regularity.
More precisely, their respective visibility ratios
obey qualitative deterministic laws,
see Figure~\ref{f:vis}a).
On the other hand these scenes are constructed
using some random process,
which means instances can be made up of any desired size $n$.

Specifically all scenes arise from throwing 
into each square of an $\sqrt{n}\times\sqrt{n}$ grid
one randomly oriented segment.
For Scene~A, these segments are then shrunk by a
factor $\cong\sqrt{n}$ to yield an average visibility count
proportional to $n$, see Figure~\ref{f:vis}a).
For Scene~B, each segment sequentially is grown as to
just touch some other one: remember we want to comply 
with Definition~\ref{d:Nondegenerate};
as expected (and corresponding to Lemma~\ref{l:VSPsize}b)
this results in constant visibility counts.
Scene~C finally arises from Scene~B 
by shrinking the segments again;
here the visibility count grows roughly proportional 
to $n^{0.3}$.

\begin{figure}[htb]
\hspace*{-2ex}%
\includegraphics[width=0.51\textwidth,height=51ex
]{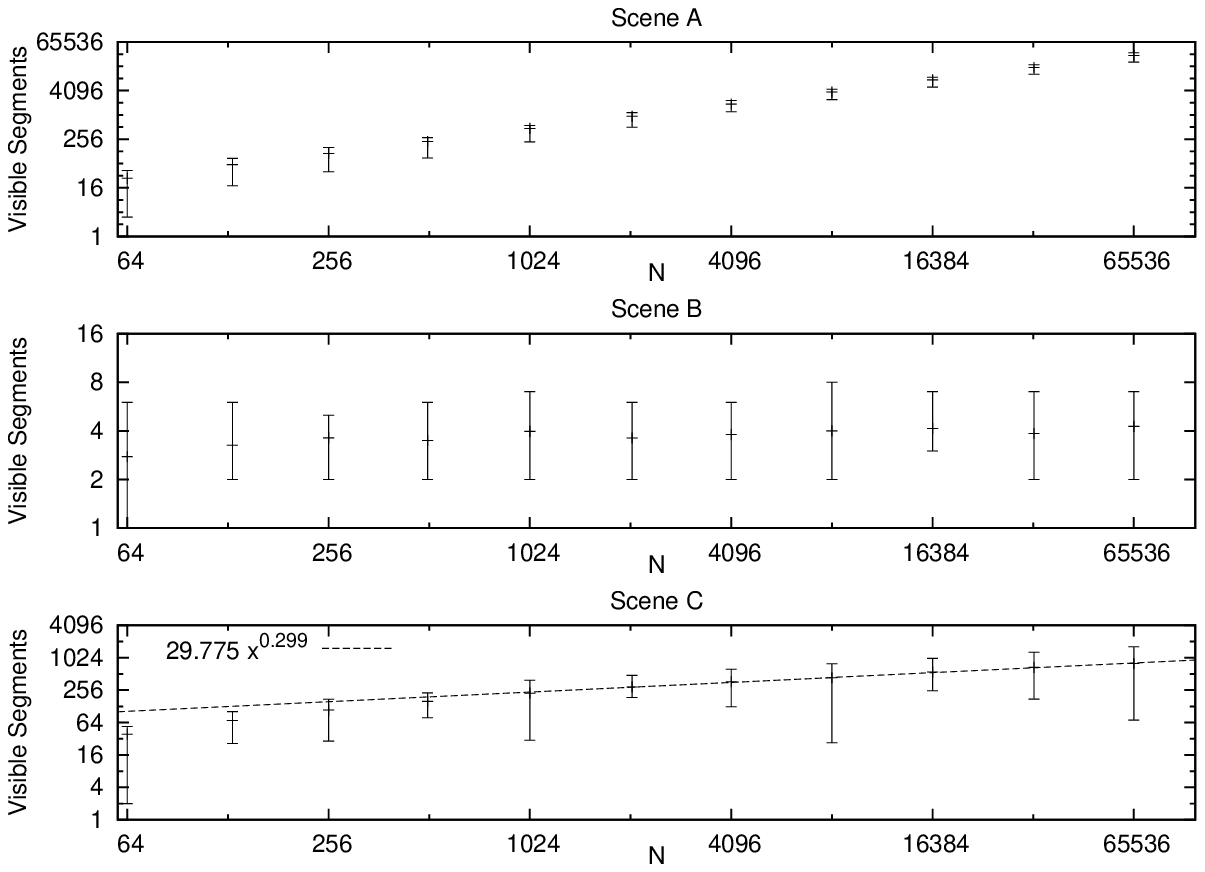}\hfill%
\includegraphics[width=0.51\textwidth,height=51ex
]{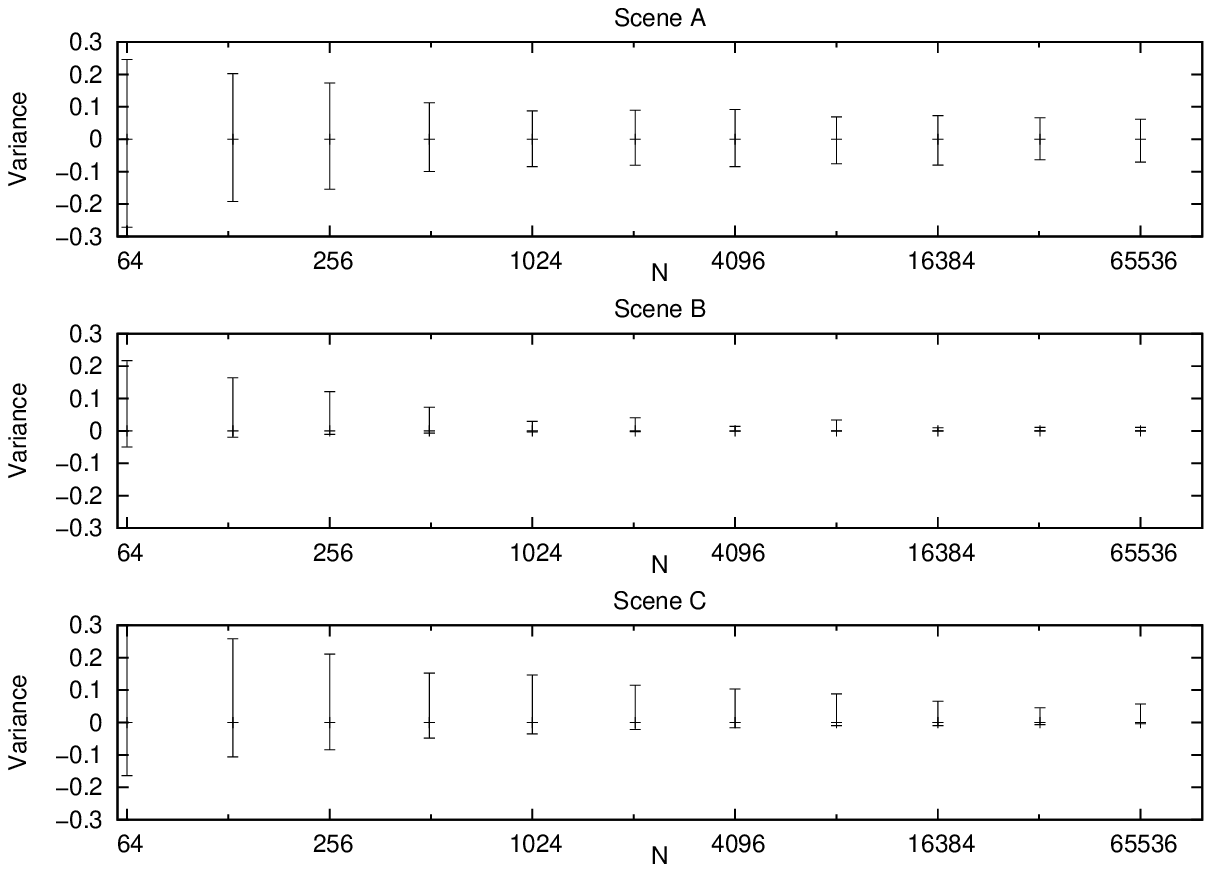}\hspace*{-1ex}%
\caption{\label{f:vis}a) Maximum/average/minimum
visibility counts for Scenes~A to C.\newline
b) Variance of the output of the randomized algorithm.}
\end{figure}

Figure~\ref{f:vis}b) indicates the quality of approximation
attained by Corollary~\ref{c:Tradeoff}: Recall that the
preprocessing step randomly selects $k$ elements of the scene 
as targets; and the proof shows that \emph{asymptotically}
this sample is `representative' for the entire scene with
high probability. Our implementation chose $k=10\cdot\log^2(n)$
which turned out to yield a practically good approximation indeed.
More precisely, Figure~\ref{f:vis}b) displays the $1\sigma$
confidence interval estimates for scenes~A) and C) from
various viewpoints, normalized to (i.e. after subtracting the)
mean 0.

\subsection{Memory Consumption versus Query Time} \label{s:MemoryTime}
Preprocessing space and query time are two major resource 
contraints for many applications such as the one we aim at.
We have thus performed extensive measurements of these 
quantities for the scene types A) to C) mentioned above.
It turns out that for A) our data structure takes
roughly linear space and for B) roughly quadratic one,
whereas for C) it grows strictly stronger;
see Figure~\ref{f:MemoryTime}a).
Here we refer to the setting $\ell=1$.
For scene~C) we have additionally employed the trade-off 
featured by Corollary~\ref{c:Tradeoff} to reduce the
memory consumption at the expense of query time;
specifically, scene~E) means scene~C) with 
$\ell:=n^{1/4}$, and scene~D) refers to 
$\ell:=n^{1/2}$. It turns out that the latter
effectively reduces the size to quadratic,
see Figure~\ref{f:MemoryTime}a).

\begin{figure}[htb]
\hspace*{-2ex}%
\includegraphics[width=0.51\textwidth]{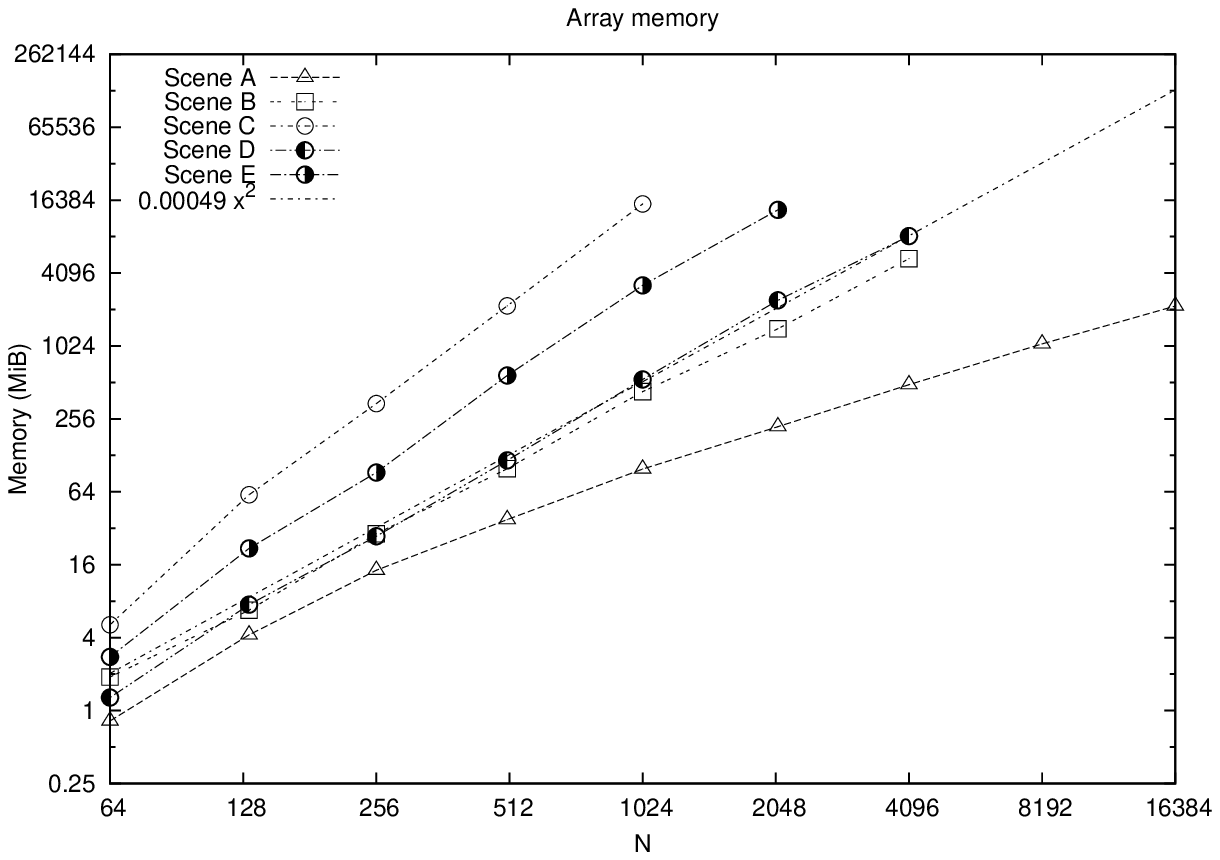}\hfill%
\includegraphics[width=0.51\textwidth]{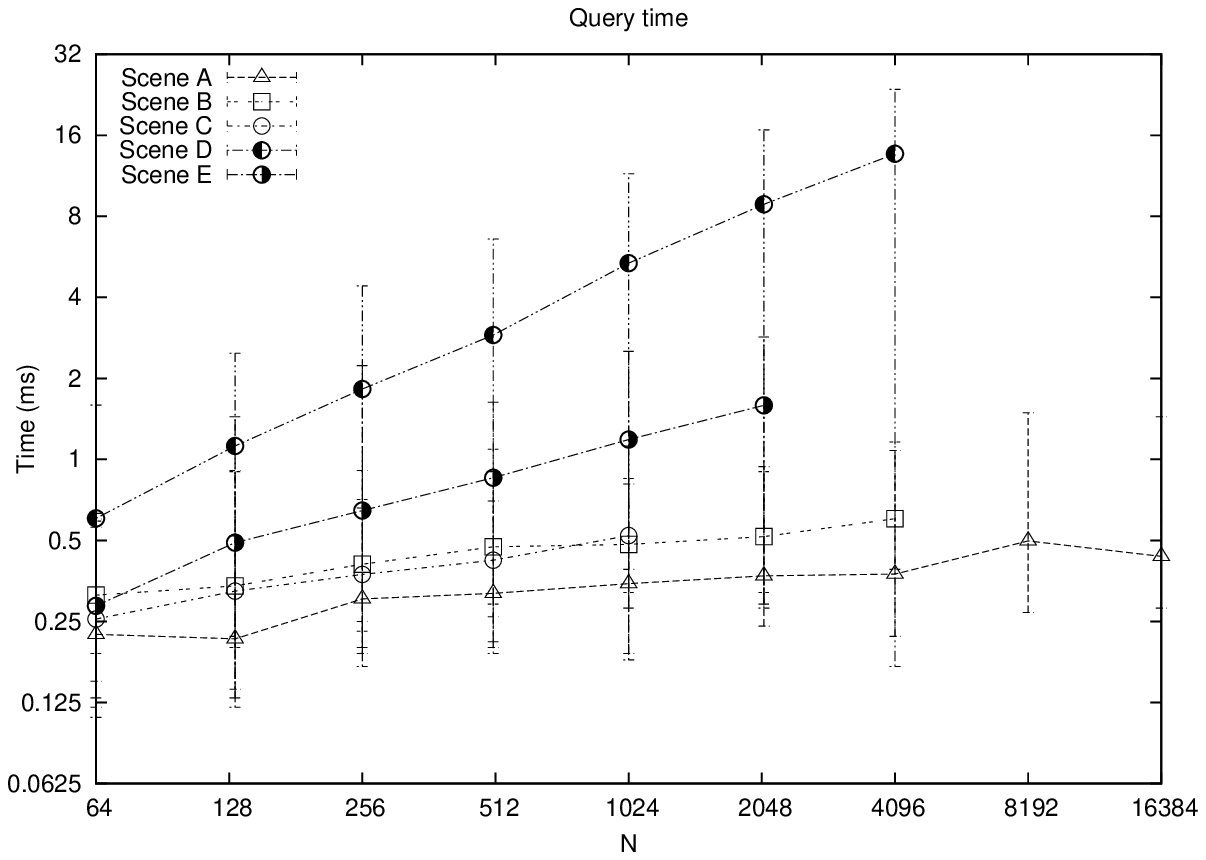}\hspace*{-1ex}%
\caption{\label{f:MemoryTime}
Log-log plot of a) preprocessing memory 
and of b) query time for scenes~A) to E).}
\end{figure}

On the other hand, the saved memory is paid for by
an increase in query time; cf. Figure~\ref{f:MemoryTime}b).
Indeed at scenes of size $n\approx20,000$ each query is estimated
to take about $50\text{ms}$, that is as long as one frame may last
at an interactive rate of $20\text{fps}$. Whereas scene~E),
that is scene~C) with $\ell=n^{1/4}$ instead of $\ell=n^{1/2}$,
is estimated to still remain far below this limit for much, much
larger scenes $n$; not to mention scenes~A) to C), i.e.
with $\ell=1$.

\subsection{Conclusion}
Our benchmarks range up to $n\approx 8,000$
when the data structure hit an overall memory
limit of 16GB.
This may first seem to fall far short of the original sizes
aimed at in Section~\ref{s:Culling}.
On the other hand,
\begin{itemize}\itemsep0pt%
\item the measurements obtained turn out to depend
smoothly on $n$ and thus give a sufficient indication of, and permit
to extrapolate with convincing significance, the behavior
on larger scenes.
\item By proceeding from single geometric simplicies to entire 
`virtual objects' (like e.g. a house or a car) as rendering primitives,
one can in practice easily save a factor of 100 or 1000.
\item Temporal and spacial coherence of an observer moving within
a virtual scene suggests that visibility counting queries need
not be performed in each frame separately. 
Moreover our CPU-based algorithm can be run 
concurrently to the graphics processing unit (GPU).
These two improvements in running time can then be
traded for an additional saving in memory.
\item In order to attain and access the 16GB mentioned above,
we employed secondary storage (a harddisk):
with an unfavorable increase in preprocessing time
but suprisingly little effects on the query time.
\end{itemize}
These observations suggest that our algorithm's practicality
can be extended to $n$ much larger than the above $8,000$;
yet doing so is beyond the purpose of the present work.

It is thus fair to claim as
main benefit of our contribution in Corollary~\ref{c:Tradeoff}
an (as opposed to Corollary~\ref{c:Coarse}) practically relevant
approach to approximate visibility counting based on the ability
to trade (otherwise prohibitive quartic) preprocessing space for
(otherwise almost neglectible logarithmic) query time.

\section{Perspectives} \label{s:Perspective}
\begin{enumerate}
\item[a)]
We have treated the observer's position $\vec x$ as an input
newly given from scratch for each frame. 
In practice however $\vec p$ is more likely to move continuously
and with bounded velocity through the scene. 
This should be exploited algorithmically, e.g. in form of a
\textsf{visibility \emph{count} maintenance} problem \mycite{Section~3.2}{Flatland2}.
\item[b)]
How does Theorem~\ref{t:TimeVersusSpace}
extend from 2D to 3D,
what is the typical size of a 3D VSP?
\item[c)]
The quartic worst-case size of 2D VSPs (and quadratic typical
yet even of order $n^9$ for 3D) arises from visibility considered
with respect to \emph{perspective} projections; whereas for
\emph{orthographic} projections, it drops to $\calO(n^2)$
(in 2D; in 3D: order $n^6$) \cite{AspectSurvey}.
\item[d)] 
The counterexamples in Lemma~\ref{l:VSPsize}a) 
and Lemma~\ref{l:VCdim}d) and also Proposition~\ref{p:Coarse}a) 
employ (after scaling the
entire scene to unit size) very short and/or
very close segments. We wonder if such worst cases can be
avoided in the \emph{bit} cost model, i.e. with respect
to $n$ denoting the total binary length of the
scene description on an integer grid.
\COMMENTED{
\item[e)]
We have evaluated the benefit of visibility numbers 
as a guide for 
when to apply occlusion culling and when to turn it off. 
It remains to devise and analyze 
culling algorithms which 
permit, through an additional input parameter $u$
\emph{between} 0 (brute-force: render everything)
and 1 (render precisely those objects at least partly visible).
\item[f)]
Controlling occlusion culling according to visibility counts
amounts to \emph{input-}adaptive algorithms (Section~\ref{s:Adaptive}).
It seems equally important and promising to attack
\emph{hardware-}adaptivity (recall Section~\ref{s:Unrelated})
and, specifically, to take into account the semantics
and asynchronity of contemporary \texttt{GL\_OCCLUSION} tests.
}
\end{enumerate}

\subsection{Remarks on Lower Bounds}
We have presented various data structures and algorithms 
for visibility counting, trading
preprocessing space for query time.
It would be most interesting to complement these results
by corresponding lower bounds of the form:
preprocessing space $s$ requires,
in some appropriate model of computation, 
query time at least $\Omega\big(f(s)\big)$.
Unfortunately techniques for Friedman's \textsf{Arithmetic Model}, 
which have proven so very successful for range query problems
\cite{LowerBound}, do not apply to the non-faithful
semigroup weights of geometric counting problems,
not to mention decision Problem~\ref{p:Visibility};
and approximate, rather than exact, counting
makes proofs even more complicated.

On the other hand, we do have some lower bounds:
namely on the sizes of VSPs and relaxed VSPs
in Lemma~\ref{l:VSPsize}a) and Proposition~\ref{p:Coarse}a+b).
These immediately translate to lower bounds on the sizes
of \textsf{Linear Branching Programs},
based on the observation that any such
program needs a different leaf for each different
convex cell of inputs leading to the same output value,
compare \cite{Vio}.

But then again this seemingly natural model of computation
is put into question when considering the algorithms
from Sections~\ref{s:IntervalTree} and \ref{s:RotSweep}:
Both have, in spite of Lemma~\ref{l:VSPsize}a), 
merely (weakly) linear size. Superficially,
Lemma~\ref{l:RotSweep} seems to employ transcendental
functions for angle calculations; but these can be
avoided by comparing slopes instead of angles---however
employing divisions. These, again, can be replaced:
yet multiplications do remain and make the algorithm
inherently \emph{non}linear a branching program.

\subsection{Visibility in Dimensions $>2$} \label{s:3SUM}
We had deliberately restricted to the planar case of line segments.
Many virtual scenes in interactive walkthrough applications
can be described as $2\tfrac{1}{2}$-dimensional:
buildings of various heights yet rooted on a common plane.

But how about the full 3D case?
Here we observe a quadratic 
`almost' lower bound on the joint running time
of preprocessing and querying for the 3D
counterpart to Problem~\ref{p:Visibility}.
To this end recall the following famous

\begin{problem}[3SUM]
Given an $n$-element subset $S$ or $\IR$,
do there exist $a,b,c\in S$ such that $a+b+c=0$ ?
\end{problem}
It admits an easy algebraic $\calO(n^2)$-time algorithm
but is not known solvable in subquadratic time. Similar
to Boolean Satisfiability (\textsc{SAT}) and the theory
of $\calNP$-completeness, \textsc{3SUM} has led to a 
rich family of problems mutually reducible one to another
in softly linear time $\calO(n\cdot\polylog n)$ and hence
called 3SUM-\emph{complete}; for example
it holds \mycite{Section~6.1}{Gajentaan}:

\begin{fact}
Given a collection $\calS$ of opaque horizontal triangles in space,
one further horizontal triangle $T$, and a viewpoint $\vec p\in\IR^3$.
The question of whether some point of $T$ is visible from $\vec p$
through $\calS$ (called \textsf{Visible-Triangle})
is 3SUM-complete.
\end{fact}
In particular there is no 3D counterpart to the Interval Tree
solving the corresponding 2D problem in time $\calO(n\cdot\log n)$,
recall Section~\ref{s:IntervalTree}.


\end{document}